\documentclass[sigconf,10pt]{acmart}

\usepackage[english]{babel}
\usepackage{amsmath,amssymb,amsfonts}
\usepackage{mathtools}
\usepackage{booktabs}
\usepackage{balance}
\usepackage{bm}
\usepackage{empheq}
\usepackage{subcaption}
\usepackage{url}
\usepackage{hyperref}
\usepackage{color,soul}

\usepackage{enumitem}
\setlist[enumerate]{topsep=0pt,itemsep=-1ex,partopsep=1ex,parsep=1ex,leftmargin=2ex}
\setlist[itemize]{topsep=0pt,itemsep=-1ex,partopsep=1ex,parsep=1ex,leftmargin=2ex}

\renewcommand\footnotetextcopyrightpermission[1]{} 
\acmYear{2027}
\copyrightyear{2027}
\setcopyright{rightsretained}
\acmConference[ACM MobiCom '27]{MobiCom 2027}{October 2027}{}
\acmBooktitle{ACM MobiCom '27, October 2027}

\title{Who Can Safely Wait? Risk-Aware 5G Scheduling for\\ Low-Altitude UAV Operations}

\begin{document}

 \author{André Melo}
 \affiliation{%
   \institution{Aeronautics Institute of Technology}
   \city{São José dos Campos}
   \country{Brazil}}
 \email{andre.melo@ita.br}

 \author{Yueqi Chen}
 \affiliation{%
   \institution{Northwestern University}
   \city{Evanston}
   \state{IL}
   \country{USA}}
 \email{yueqichen@northwestern.edu}

 \author{Lourenço Pereira Jr.}
 \affiliation{%
   \institution{Aeronautics Institute of Technology}
   \city{São José dos Campos}
   \country{Brazil}}
 \email{ljr@ita.br}

 \author{Igor Kadota}
 \affiliation{%
   \institution{Northwestern University}
   \city{Evanston}
   \state{IL}
   \country{USA}}
 \email{kadota@northwestern.edu}

\renewcommand{\shorttitle}{Risk-Aware 5G Scheduling for Low-Altitude UAV Operations}
\settopmatter{authorsperrow=4}

\begin{abstract}
Cellular-connected UAV operations already fly in many countries. However, cellular networks were designed for terrestrial users and remain blind to the operational consequences of delaying information from low-altitude aircraft. We present Risk-Aware (RA), a framework that translates information freshness and operational context into aircraft-specific operational risk, enabling uplink scheduling based on operational consequence rather than network metrics alone. RA balances the consequence of further delay against the service cost of heterogeneous single- and multi-slot updates. It combines a regulation-inspired risk proxy, risk-directed prioritization based on marginal risk growth, and a Max-Weight scheduler with Lyapunov-based performance guarantees. Although RA is derived from an analytically tractable model, we evaluate it under moving UAV trajectories, dynamic operational context, geographically varying population exposure, and realistic UAV video traffic. Across 48 operating points and fleets of up to 120 UAVs in Sao Paulo delivery operations and the Interlagos Formula 1 Grand Prix setting, RA reduces 99th-percentile per-UAV risk by up to 96\% over proportional fair and 86\% over round robin, while supporting fleets up to 1.5$\times$ and 3$\times$ larger, respectively, within the same risk envelope.

\end{abstract}


\maketitle

\section{Introduction}
\label{sec:intro}

Low-altitude airspace is poised for a fundamental transformation. Airspace authorities worldwide~\cite{USDOT_AAM_National_Strategy_2025,sesarConops,dca351-6,uk_conops, uk_roadmap} are developing new traffic-management frameworks for UAVs, powered-lift vehicles, remotely piloted aircraft, and increasingly autonomous vehicles operating alongside traditional aviation. 
Existing air traffic infrastructure, built largely around voice communications and surveillance (radar), was not designed for the anticipated density, heterogeneity, or automation of low-altitude aviation \cite{faaConops,USDOT_AAM_National_Strategy_2025}. This transformation will require modern communication~\cite{USDOT_AAM_National_Strategy_2025} for timely information exchange among aircraft, operators, and infrastructure. 


These traffic-management frameworks~\cite{USDOT_AAM_National_Strategy_2025,sesarConops,dca351-6,uk_conops, uk_roadmap} will rely on timely exchange of aircraft status, telemetry, intent, airspace constraints, and airspace traffic information. Remotely supervised and autonomous operations may also generate video and other high-volume sensing streams~\cite{ahmad2023aerotraj,alipour2024datmo,kafetzis2025} for remote situational awareness. The Specific Operations Risk Assessment (SORA)~\cite{sora2024,jarus_sora25_annexf}, a risk-assessment methodology for UAV operations, classifies operations into a Specific Assurance and Integrity Level (SAIL) class I-VI that sets the rigor of the required mitigations, and requires the Command, Control, and Communication (C3) link to have performance and spectrum characteristics appropriate for the operation. For example, flights over sparsely populated ground receive lower SAIL classes, whereas flights over crowds receive higher classes.  


Commercial cellular networks are a promising component of this communication architecture~\cite{USDOT_AAM_National_Strategy_2025, USDOT_AAM_Comprehensive_Plan_2025} because they provide widespread infrastructure, high-capacity uplinks, dynamic resource allocation, and integration with edge computing. 
Cellular-connected UAV operations are already authorized or active in several countries~\cite{meituan,keeta_dubai,anac_speedbird}, shifting the central question from whether cellular networks can connect aircraft to whether they can allocate limited radio resources consistently with operational safety~\cite{USDOT_AAM_National_Strategy_2025, USDOT_AAM_Comprehensive_Plan_2025}.

Cellular connectivity alone, however, does not make the network suitable for airspace management. Under contention, the Radio Access Network node (also called Next Generation Node B, or gNB) must decide which aircraft and information streams to serve and which to defer. Conventional schedulers make this decision using network-centric metrics such as Channel Quality Indicator (CQI), Buffer Status Reports (BSRs), throughput, fairness, and traffic class. These metrics characterize the opportunity and cost of transmitting data, but not the operational consequence of delaying it. As a result, a scheduler may improve every network metric it observes while withholding resources from the aircraft for which stale information is most consequential.


Figure~\ref{fig:blindspot} exposes this mismatch in a minimal six-UAV
  corridor where radio quality and operational criticality vary in opposite directions. Uplink SNR decreases toward a crowded destination, so the segment near the gNB corresponds to a SAIL II operation while over the crowd corresponds to SAIL V. Every UAV carries the same 30 fps video and 10 Hz status traffic (described in \S~\ref{sec:eval-methodology}). Best-CQI maximizes network throughput by concentrating
  service near the gNB and repeatedly deferring the UAV flying
  over the crowd: it schedules that aircraft $4\times$ less often than
  Round-Robin (Fig.~\ref{fig:blindspot}b), and its packets wait $8\times$
  longer (Fig.~\ref{fig:blindspot}c). Clearly, throughput maximization comes at the cost of increased operational risk. 

\begin{figure*}[t]
        \includegraphics[width=\textwidth]{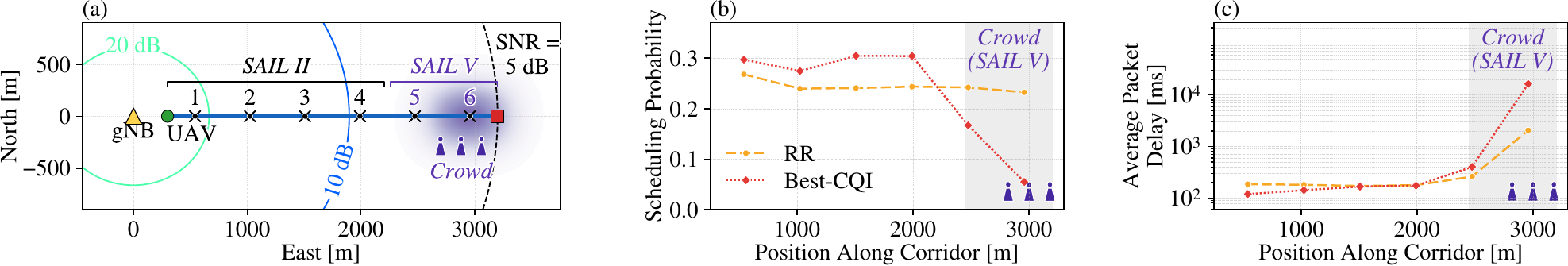}
        \vspace{-0.3cm}
    \caption{The scheduling blind spot using an OpenAirInterface (OAI) gNB emulator with a complete 3GPP MAC pipeline.
  (a)~Six UEs emulate static UAVs along a corridor whose uplink channel degrades from the gNB toward a crowded event; segments are labeled with the SAIL risk class (II over open ground, V over the crowd).(b)~Scheduling probability,
  the fraction of uplink slots in which each aircraft is scheduled (up to $4$ 
  UEs can share a slot). 
  (c)~Average packet delay of status updates and
  video frames, in log scale. 
  }
\label{fig:blindspot}
\vspace{-0.5cm}
\end{figure*}

The missing operational context matters under contention because equal communication delays can have different physical consequences. Stale position information from an UAV over an unpopulated corridor may consume little operational risk margin, while the same staleness near another aircraft or above a dense crowd may produce substantially greater risk. The scheduler must also account for heterogeneous transmission costs: an aircraft-status update may require a single transmission opportunity, whereas a video frame may require many. Resource allocation must therefore balance the operational value of transmitting an update against the communication resources required to deliver it. Accounting for this balance can reduce communication-induced operational risk for a given fleet (\S~\ref{sec:eval-risk}) or support a larger fleet within a fixed operational risk envelope (\S~\ref{sec:eval-capacity}).

Existing network mechanisms do not account for this context. Slicing and reservation can isolate aviation traffic or provision capacity for a traffic class \cite{foukas2017orion,zhang2026smec}, but they do not determine which aircraft or information stream is most costly to defer. Channel-aware and proportional-fair schedulers prioritize radio efficiency or balance throughput across users, without accounting for differences in operational risk. Age of Information (AoI) scheduling prioritizes stale information \cite{huang2015,kaul2012real,sun2017,zhao2024large,zhao2025}, but staleness alone does not identify where freshness is most valuable. 

This exposes a fundamental blind spot: cellular schedulers observe network state, while the operational consequences of their decisions depend on context outside the network’s decision model. Bridging this gap requires an interface that translates network state and operational context into a quantity the scheduler can evaluate. Such an interface would allow the network to compare not only the efficiency of serving each stream, but also the operational cost of making it wait.

Emerging 5G capabilities make such an interface increasingly feasible. The gNB already maintains per-UE radio and traffic state through CQI, BSRs, HARQ feedback, and delivery history. Packet timestamps and delivery records expose the staleness of aircraft-status, telemetry, and sensing streams. Aircraft position can be obtained from onboard reports or the 5G Location Management Function (LMF), while the Network Data Analytics Function (NWDAF) and external geographic data can provide estimates of population exposure and other operating conditions \cite{ts23273,ts23288}. Operational context, such as aircraft location and population exposure, evolves slowly over seconds or minutes. The gNB combines the operational context with per-slot network metrics to make scheduling decisions at millisecond timescales. Although positioning and analytics remain limited in accuracy and update rate \cite{3gpp_tr38857,lin2024camp}, they provide the information needed to make uplink scheduling operationally aware.

Using this information effectively requires a scheduler to satisfy three design requirements. First, it must \textbf{prioritize marginal operational risk}, directing service toward the aircraft for which an additional unit of delay is most consequential. Second, it must balance this urgency against \textbf{heterogeneous service costs}, since a status update may require one transmission opportunity while a video frame may require many. Third, it must be \textbf{deployable within the existing 5G architecture}. In particular, it must convert slowly evolving operational context into priorities that can be evaluated within the gNB’s per-slot scheduling loop.

This paper asks a simple question: \textbf{When a cellular network cannot immediately serve every aircraft, who can safely wait?} We answer this question with Risk-Aware (RA) scheduling, a framework that makes operational risk the organizing principle for 5G uplink resource allocation. RA prioritizes the aircraft and information streams for which further delay is most consequential, while accounting for delivery probability and service cost. The resulting policy operates entirely at the gNB, requiring no modification to the UE or air interface. RA also relates the resulting operational-risk envelope to the number of aircraft supportable by the available wireless resources, connecting radio scheduling to operational capacity. The problem formulation and evaluation scenarios incorporate input from air-traffic operations, airspace management, and safety-assessment practices. We make the following contributions:


    \noindent\textbf{(i) A network-computable operational-risk abstraction.} We translate gNB-observable network metrics and operational context into aircraft-specific operational risk. Drawing on the principles underlying the SORA~\cite{sora2024,jarus_sora25_annexf}, our model captures three consequences of stale information: position uncertainty consuming aircraft-separation margin, population exposure within the position uncertainty region, and loss of coherence across recently delivered video frames. 
    
    \noindent\textbf{(ii) A Risk-Aware Max-Weight scheduler for heterogeneous updates.} We formulate the joint scheduling of UAV-status updates and video frames as a long-term operational-risk minimization problem. We derive a universal lower bound over feasible status-delivery and frame-completion rates, then design a Lyapunov function that accounts for video staleness, transmission progress, and delayed frame completion. The resulting RA Max-Weight policy balances marginal operational-risk reduction, packet-delivery probability, and service cost using a low-complexity per-slot rule. We establish a worst-case optimality ratio of four.
    
    \noindent\textbf{(iii) Validation and evaluation grounded in realistic operating environments}. We develop a scalable slot-level MATLAB simulator for evaluating 5G uplink scheduling in low-altitude UAV operations. The simulator is cross-checked against UAV Toolbox~\cite{toolboxuav} mobility and driven by 5G Toolbox~\cite{toolbox5g} propagation models, while capturing realistic UAV traffic generation from~\cite{aviator2021}. We construct controlled scenarios grounded in real S\~ao Paulo geographies, licensed 5G deployments, documented UAV delivery operations, population maps, and a mass-gathering setting at the Interlagos Formula 1 circuit. Because this abstraction enables broader evaluation than is practical with a full 5G stack, we validate it against the OAI implementation used in Fig.~\ref{fig:blindspot} under the proportional-fair experiments described in \S~\ref{sec:eval-methodology}. We evaluate RA and baseline policies across 48 operating points with fleets of up to 120 UAVs. RA reduces 99th-percentile per-UAV operational risk by up to 96\%
        relative to Proportional Fair, 86\% relative to Round-Robin, and 83\% relative to AoI scheduling. Within the same operational-risk envelope, RA supports fleets up to $1.5\times$ larger than Proportional Fair, $3\times$ larger than Round-Robin, and \(2\)--\(3\times\) larger than AoI scheduling.    
    
    

    \noindent\textbf{(iv) Communication-constrained operational capacity.} Using our simulator, we evaluate the largest simultaneous UAV fleet that a given uplink resource budget can support while keeping every aircraft within the declared communication-induced operational-risk envelope. Unlike conventional capacity metrics based on throughput or device count, this abstraction incorporates aircraft context, the consequences of stale information, and the allocation of radio resources among heterogeneous streams. It connects millisecond-scale gNB scheduling to longer-timescale UAV admission control and traffic-flow management: when projected UAV demand approaches the operational-capacity boundary, the network can expose the need to delay UAV departures, adjust speeds, reroute aircraft, or provision additional communication resources. 

\textbf{Paper organization}. 
Section~\ref{sec:stale-to-risk} defines the system model, operational-risk abstraction, and optimization problem. 
Section~\ref{sec:policy} develops the Risk-Aware scheduler and establishes its performance guarantees. 
Section~\ref{sec:eval} evaluates operational risk, network performance, and operational capacity. 
Section~\ref{sec:related} reviews related work. 
Section~\ref{sec:conc} concludes the paper. 

\section{Network Model and Operational Risk}
\label{sec:stale-to-risk}
The Specific Operations Risk Assessment (SORA), which underpins contemporary UAV regulations and standards \cite{ed253,ed267,ed272,faaConops,sesarConops}, characterizes an
operation through its final Ground Risk Class (GRC) and residual Air Risk Class (ARC). Their combination determines a SAIL Level (SAIL I--VI), which in turn specifies the robustness required for the applicable Operational Safety Objectives
(OSOs) \cite{sora2024}. 

Relevant to our setting, OSO~\#06 requires that the performance of the C3 links is appropriate for the intended operation. SORA, however, does not specify how a cellular scheduler should allocate resources among simultaneous C2, telemetry, sensing, and video transmissions, nor does it define per-flow scheduling priorities~\cite{sora2024}.

Yet, SORA provides the principles needed to close this gap. Its risk classes relate operational criticality to population, air traffic, obstacles, and uncertainty \cite[\S~1.4]{jarus_sora25_annexf}. For example, the intrinsic Ground Risk Class (iGRC) combines population density with an uncertainty footprint, whose radius is derived from the aircraft's maximum speed. By combining these principles with network-induced information staleness, location, and flow type, we enable the 5G scheduler to infer which stale update is consuming the operational risk envelope the fastest. 

Next, we describe our network model; then we relate the network performance to operational risk. Finally, we formulate the 5G scheduling optimization problem.

\subsection{Network Model}
\label{sec:network-model}
Consider a single-cell 5G NR uplink system with \(N\) UAVs communicating with a ground base station (gNB), with their uplink traffic isolated from terrestrial-user traffic through network slicing. Time is slotted and indexed by \(t\in\{1,2,\ldots\}\). 
At each slot, the gNB scheduler allocates at most \(K\) transmission opportunities to the UAV slice.
Each UAV has two uplink information flows: a C2 status-update flow and a video-frame flow. 
Both flows follow a generate-at-will model: a fresh status update is generated whenever
the UAV is scheduled for status transmission, and a fresh video
frame is generated whenever the UAV is scheduled to start a
video frame transmission.\footnote{Although status updates and video
frames are typically generated periodically in practice, the adopted
models approximate high-rate status sensing and continuous video
streaming. They also allow us to focus on the effects of scheduling
and transmission reliability without introducing additional sampling
parameters.}
Each status update consists of a single packet and can be transmitted within one slot,
whereas each video frame consists of \(L\) packets and is
considered successfully delivered only after all \(L\)
packets have been successfully received. 
We assume\footnote{UAV positions can be obtained through network-assisted or
onboard localization \cite{ts23273}. Population
exposure can be estimated from geographic population data, while
staleness can be tracked using packet timestamps and delivery records
\cite{zhao2024large}.} that the cellular infrastructure provides the gNB with estimates of each UAV's position, the population density, and the staleness of its information flows.

\noindent\textbf{Packet Transmission.}
In each slot \(t\), the scheduler decides which UAVs will attempt to transmit.
Let
\(a_{u,s}(t)\in\{0,1\}\)
denote the status packet scheduling indicator, where
\(a_{u,s}(t)=1\) if the status packet of UAV \(u\) is scheduled
in slot \(t\), and
\(a_{u,s}(t)=0\) otherwise.
Similarly, let
\(a_{u,v}(t)\in\{0,1\}\)
denote the video packet scheduling indicator, where
\(a_{u,v}(t)=1\) if one video packet of UAV \(u\) is scheduled
in slot \(t\), and
\(a_{u,v}(t)=0\) otherwise.
The scheduling decisions satisfy
\begin{equation}
\sum_{u=1}^{N}
\bigl(
a_{u,s}(t)+a_{u,v}(t)
\bigr)
\le K,
\qquad
\forall t,
\label{eq:resource_constraint}
\end{equation}
meaning that at most \(K\) transmission opportunities can be
scheduled in each slot.

A scheduled status packet is successfully delivered with
probability \(p_{u,s}\).
Let
\(d_{u,s}(t)\in\{0,1\}\)
denote the successful status packet transmission indicator, where
\(d_{u,s}(t)=1\) if the scheduled status packet is successfully
delivered in slot \(t\), and
\(d_{u,s}(t)=0\) otherwise.
Then \(\mathbb E
[d_{u,s}(t)\mid a_{u,s}(t)]
=
p_{u,s}a_{u,s}(t).\)
Similarly, a scheduled video packet is successfully delivered
with probability \(p_{u,v}\).
Let
\(\mu_{u,v}(t)\in\{0,1\}\)
denote the successful video packet transmission indicator,
where
\(\mu_{u,v}(t)=1\) if the scheduled video packet is successfully
delivered in slot \(t\), and
\(\mu_{u,v}(t)=0\) otherwise.
Then \(\mathbb E
[\mu_{u,v}(t)\mid a_{u,v}(t)]
=
p_{u,v}a_{u,v}(t).\)

\noindent\textbf{Video Transmission Process.}
Each video frame consists of \(L\) packets. At most one packet
of the current frame can be transmitted per slot, and its packets
need not be transmitted in consecutive slots. Once all \(L\) packets have been successfully received, the frame is completed.

Let \(\ell_u(t)\in\{1,\ldots,L\}\) denote the number of remaining packets of the current video frame of UAV \(u\) at the beginning of slot \(t\).
Initially, \(\ell_u(1)=L\) and each successful video packet transmission decreases \(\ell_u(t)\) by one.
After the last remaining packet is successfully delivered, \(\ell_u(t)\) is reset to \(L\) in the next slot.
Accordingly, the evolution of \(\ell_u(t)\) is given by
\begin{equation}
\ell_u(t+1)
=
\begin{cases}
L,
&
\ell_u(t)=1,\;
\mu_{u,v}(t)=1,
\\
\ell_u(t)-1,
&
\ell_u(t)>1,\;
\mu_{u,v}(t)=1,
\\
\ell_u(t),
&
\mu_{u,v}(t)=0.
\end{cases}
\label{eq:remaining_packet}
\end{equation}

Let \(g_u(t)\) denote the age of the current video frame at UAV
\(u\), measured at the beginning of slot \(t\) as the elapsed
time since the frame was generated. When no video-frame
transmission is in progress, we set \(g_u(t)=0\). Initially,
\(g_u(1)=0\). The evolution of \(g_u(t)\) is given by
\begin{equation}
g_u(t+1)
=
\begin{cases}
0,
&
\ell_u(t)=1,\;
\mu_{u,v}(t)=1,
\\
0,
&
\ell_u(t)=L,\;
\mu_{u,v}(t)=0,
\\
g_u(t)+1,
&
\textnormal{otherwise}.
\end{cases}
\label{eq:system_time}
\end{equation}
Under the generate-at-will model, when a UAV is scheduled for
video transmission with \(\ell_u(t)=L\), it generates a fresh
frame at the beginning of the slot. If the first packet is not
successfully delivered, the frame is discarded, and another
fresh frame is generated at the next video scheduling
opportunity. Once the first packet is successfully delivered,
the same frame is retained until all its remaining packets are
delivered.
The video frame delivery indicator is defined as \(d_{u,v}(t)
=
\mathbf 1\{\ell_u(t)=1,\;\mu_{u,v}(t)=1\},\)
where
\(d_{u,v}(t)=1\)
if the current video frame is successfully delivered in slot
\(t\), and
\(d_{u,v}(t)=0\)
otherwise.

\noindent\textbf{AoI Evolution.}
Let \(h_{u,s}(t)\) denote the AoI at the gNB of the most recently
successfully received C2 status update from UAV \(u\) at the
beginning of slot \(t\). We refer to \(h_{u,s}(t)\) as the status AoI. The evolution of \(h_{u,s}(t)\) is given by
\begin{equation}
h_{u,s}(t+1)
=
\begin{cases}
1,
&
d_{u,s}(t)=1,
\\
h_{u,s}(t)+1,
&
\textnormal{otherwise}.
\end{cases}
\label{eq:status_aoi}
\end{equation}
Let \(h_{u,v}(t)\) denote the AoI at the gNB of the most recently
completed video frame from UAV \(u\) at the beginning of slot \(t\).
We refer to \(h_{u,v}(t)\) as the video AoI.
Upon a successful video frame delivery in slot \(t\), the video AoI is reset to the age of the delivered frame at the
beginning of slot \(t+1\), namely \(g_u(t)+1\).
Therefore,
\begin{equation}
h_{u,v}(t+1)
=
\begin{cases}
g_u(t)+1,
&
d_{u,v}(t)=1,
\\
h_{u,v}(t)+1,
&
\textnormal{otherwise}.
\end{cases}
\label{eq:video_aoi}
\end{equation}
We assume that \(h_{u,s}(1)=h_{u,v}(1)=1\) for all
\(u\in\{1,\ldots,N\}\).
\noindent\textbf{Scheduling Policy.}
A scheduling policy determines the actions
\(\{a_{u,s}(t),a_{u,v}(t)\}_{u=1}^{N}\) in each slot.
We consider non-anticipative policies that do not use future channel outcomes
or future system states when making scheduling decisions. Let
\(\Pi\) denote the class of admissible non-anticipative scheduling
policies satisfying the resource constraint
\eqref{eq:resource_constraint}.

\subsection{Staleness-Induced Uncertainty}
\label{sec:freshness-uncertainty}
Let \(h_{u,s}(t)\), defined in \eqref{eq:status_aoi}, denote the
AoI of the most recently received C2 status update from UAV \(u\)
at the beginning of slot \(t\). Let
\(\hat{\bm{x}}_u(t)\in\mathbb{R}^{2}\) denote the position
contained in that update. During the interval, the aircraft may have moved away from its last
reported position \(\hat{\bm{x}}_u(t)\).

Following the containment principle used by SORA, we assume that
a UAV can travel by at most \(v_{\max}\) per slot. We define its worst-case 
position uncertainty radius\footnote{The circular region captures worst-case position uncertainty when
no directional information is available. The framework also
supports other uncertainty models by replacing the disk with an
appropriate region. For example, a trajectory-aligned ellipse can
represent uncertainty along a preplanned path, with population
exposure evaluated over the resulting region.} as
\begin{equation}
r_u(t)
=
r_{\min}
+
v_{\max}h_{u,s}(t),
\label{eq:uncertainty-radius}
\end{equation}
where \(r_{\min}\) represents the minimum localization uncertainty. The corresponding staleness-induced uncertainty region is the
disk centered at the most recently reported position:
\begin{equation}
\mathcal{B}_u(t)
=
\left\{
\bm{x}\in\mathbb{R}^{2}:
\left\|
\bm{x}-\hat{\bm{x}}_u(t)
\right\|
\leq r_u(t)
\right\}.
\label{eq:uncertainty-region}
\end{equation}
Its area is
\begin{equation}
\left|
\mathcal{B}_u(t)
\right|
=
\pi r_u^2(t)
=
\pi
\left[
r_{\min}
+
v_{\max}h_{u,s}(t)
\right]^2.
\label{eq:uncertainty-region-area}
\end{equation}

Information staleness affects operational risk by enlarging the region in which the aircraft may be located; the resulting consequences depend on neighboring obstacles, aircraft, ground population, and other hazards.

\subsection{Operational Risk Components}
\label{sec:risk-components}
We capture three operational consequences of stale C3
information: reduced aircraft-separation margins, increased
population exposure, and degraded video coherence. Since these components measure heterogeneous quantities, we normalize them using the reference values
\(\bar R^{\mathrm{Nav}}\),
\(\bar R^{\mathrm{Ground}}\), and
\(\bar R^{\mathrm{Vid}}\), which denote the navigation, ground,
and video-coherence risk levels at their corresponding declared
operational boundaries. Their definitions are provided in
Appendix~\ref{app:risk-normalization}.

\begin{figure}[t]
    \includegraphics[width=1\columnwidth]{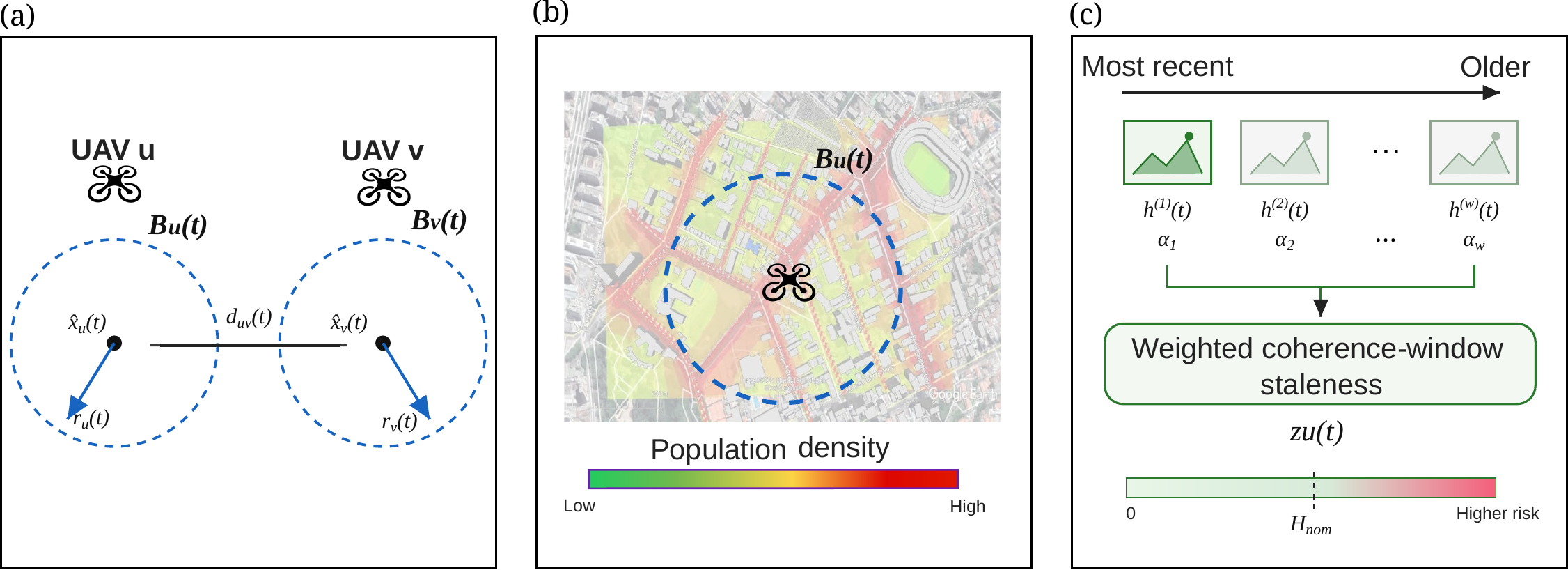}
    \vspace{-0.5cm}
    \caption{The same staleness can induce different operational consequences. (A) Navigation risk increases as position uncertainty consumes the separation margin between neighboring UAVs. (B) Ground risk depends on the population exposed within the UAV's uncertainty region, so the same AoI can imply different risk across locations. (C) Video-coherence risk captures the weighted staleness of the most recent completed frames and increases once it exceeds the nominal coherence threshold. The three components are normalized at their corresponding operational boundaries.
}
    \label{fig:risk-components}
    \vspace{-0.5cm}
\end{figure}

\noindent\textbf{Navigation Risk.}
As C2 status updates become stale, the positional uncertainties
of neighboring UAVs consume their separation margin. Let 
\(d_{uv}(t)
=
\left\|
\hat{\bm{x}}_u(t)-\hat{\bm{x}}_v(t)
\right\|\)
denote the distance between the positions contained in the most
recently received C2 status updates from UAVs \(u\) and \(v\).

Using the uncertainty radii defined in
\eqref{eq:uncertainty-radius}, we define the normalized navigation
risk of UAV \(u\) as
\begin{equation}
R_u^{\mathrm{Nav}}(t)
=
\frac{1}{\bar R^{\mathrm{Nav}}}
\sum_{\substack{v=1\\v\neq u}}^{N}
\left[
r_u(t)+r_v(t)-d_{uv}(t)
\right]^+,
\label{eq:navigation-risk}
\end{equation}
where \([x]^+\triangleq\max\{x,0\}\). The navigation risk is zero
when the reported separation exceeds the combined position
uncertainty and increases with the corresponding separation
deficit.

\noindent\textbf{Ground Risk.}
The ground risk induced by C2 status update staleness depends
on the population within the UAV's uncertainty region. Using the uncertainty region
\(\mathcal B_u(t)\) defined in \eqref{eq:uncertainty-region}, we
define the normalized ground risk as
\begin{equation}
R_u^{\mathrm{Ground}}(t)
=
\frac{1}{\bar R^{\mathrm{Ground}}}
\int_{\mathcal B_u(t)}
\rho(\bm{x},t)\,\mathrm{d}\bm{x},
\label{eq:ground-risk}
\end{equation}
where \(\rho(\bm{x},t)\) is the population-density
at location \(\bm{x}\) in slot \(t\). Thus, even for the
same C2 status update staleness, the ground risk can differ depending
on the population distribution around the UAV's most recently
reported position.

\noindent\textbf{Video-Coherence Risk.}
Remote piloting depends on a sequence of recent frames rather than
on the latest completed frame alone. Let
\(h_{u,v}^{(i)}(t)\) denote the AoI of the \(i\)-th most recently
completed video frame, where \(h_{u,v}^{(1)}(t)=h_{u,v}(t)\),
and \(h_{u,v}(t)\) is the video AoI defined in
\eqref{eq:video_aoi}. We aggregate the AoI of the \(w\) most
recently completed frames into the coherence-window staleness
\begin{equation}
z_u(t)
=
\sum_{i=1}^{w}
\alpha_i h_{u,v}^{(i)}(t),
\qquad
\sum_{i=1}^{w}\alpha_i=1,
\qquad
\alpha_i\geq 0,
\label{eq:video-coherence}
\end{equation}
where \(\alpha_i\) specifies the contribution of the \(i\)-th
most recent frame. The normalized video-coherence risk is
\begin{equation}
R_u^{\mathrm{Vid}}(t)
=
\left[
\frac{z_u(t)-H_{\mathrm{nom}}}
{\bar R^{\mathrm{Vid}}}
\right]^+,
\label{eq:video-risk}
\end{equation}
where \(H_{\mathrm{nom}}\) denotes the nominal
coherence-window staleness. Thus,
\(R_u^{\mathrm{Vid}}(t)=0\) when
\(z_u(t)\leq H_{\mathrm{nom}}\) and increases linearly with
\(z_u(t)\) otherwise.

\subsection{Composite Operational Risk}
\label{sec:composite-risk}

Using the normalized risk components defined in
\eqref{eq:navigation-risk}, \eqref{eq:ground-risk}, and
\eqref{eq:video-risk}, the per-UAV composite operational risk is
\begin{equation}
R_u(t)
=
w_N R_u^{\mathrm{Nav}}(t)
+
w_G R_u^{\mathrm{Ground}}(t)
+
w_V R_u^{\mathrm{Vid}}(t),
\label{eq:composite-risk}
\end{equation}
where \(w_N,w_G,w_V>0\) are weighting coefficients satisfying
\(w_N+w_G+w_V=1\).

The level \(R_u(t)=1\) marks the declared composite operational
boundary; values above it
indicate that communication-induced staleness has exhausted the
composite risk budget. SORA provides the framework for operational risk assessment and authorization, while this metric provides a per-slot risk proxy that allows the scheduler to compare the operational consequences of deferring different C3 streams.

\subsection{Surrogate Operational Risk}
\label{sec:surrogate-risk}

Directly optimizing the per-UAV composite operational risk in
\eqref{eq:composite-risk} is challenging because its components
are nonlinear in the AoI states and coupled through the system
dynamics. We therefore derive a linear upper bound on the per-UAV
composite operational risk.

For each UAV \(u\), the C2 status-update AoI satisfies
\begin{equation}
0
\leq
h_{u,s}(t)
\leq
\tau_{u,s}^{\max},
\qquad
\forall u,t,
\label{eq:status-age-operating-region}
\end{equation}
where \(\tau_{u,s}^{\max}\) captures the maximum tolerable
status staleness of UAV \(u\), meaning that
\(h_{u,s}(t)-\tau_{u,s}^{\max}>0\) should trigger an emergency
response. The corresponding maximum uncertainty radius of
UAV \(u\) is \(r_u^{\max}
=
r_{\min}
+
v_{\max}\tau_{u,s}^{\max}.\)
Let \(\mathcal B_u^{\max}\) denote the uncertainty region
centered at the most recently reported position of UAV \(u\)
with radius \(r_u^{\max}\). Under
\eqref{eq:status-age-operating-region}, the uncertainty region
of UAV \(u\) satisfies
\(\mathcal B_u(t)\subseteq\mathcal B_u^{\max}\).
We define the maximum population density within
\(\mathcal B_u^{\max}\) as \(\rho_{u,\mathrm{max}}
=
\max_{\bm{x}\in\mathcal B_u^{\max}}
\rho(\bm{x}).\)
It follows that
\begin{equation}
0
\leq
\rho(\bm{x},t)
\leq
\rho_{u,\mathrm{max}}
\leq
\rho_{\mathrm{max}},
\qquad
\forall u,t,\;
\bm{x}\in\mathcal B_u(t),
\label{eq:population-upper-bound}
\end{equation}
where \(\rho_{\mathrm{max}}\) is the maximum population
density within the declared operating region. Thus,
\(\rho_{u,\mathrm{max}}\) captures the local population
context associated with UAV \(u\).
For each pair of UAVs \(u\) and \(v\), let \(d_{uv}\) denote
the separation between their most recently reported positions
under the considered operational context.

The following proposition provides a linear upper bound on the
per-UAV composite operational risk within the declared operating
region.

\begin{proposition}
\label{prop:risk-upper-bound}
Under \eqref{eq:population-upper-bound} and
\eqref{eq:status-age-operating-region}, the per-UAV composite
operational risk of UAV \(u\) satisfies
\begin{equation}
R_u(t)
\leq
\alpha_{u,s}h_{u,s}(t)
+
\alpha_{u,v}h_{u,v}^{(w)}(t)
+
C_u,
\label{eq:risk-upper-bound}
\end{equation}
where
\begin{align}
\alpha_{u,s}
&=
\frac{w_N}
{\bar R^{\mathrm{Nav}}\tau_{u,s}^{\max}}
\sum_{\substack{v=1\\v\neq u}}^{N}
\left(
\left[
2r_{\min}
+
v_{\max}
\left(
\tau_{u,s}^{\max}
+
\tau_{v,s}^{\max}
\right)
-
d_{uv}
\right]^+
\right.
\notag\\
&\hspace{25mm}\left.
-
\left[
2r_{\min}
+
v_{\max}\tau_{v,s}^{\max}
-
d_{uv}
\right]^+
\right)
\notag\\
&\quad+
\frac{w_G\rho_{u,\mathrm{max}}\pi}
{\bar R^{\mathrm{Ground}}}
\left(
2r_{\min}v_{\max}
+
v_{\max}^{2}\tau_{u,s}^{\max}
\right),
\label{eq:alpha-status-definition}
\\
\alpha_{u,v}
&=
\frac{w_V}{\bar R^{\mathrm{Vid}}},
\label{eq:alpha-video-definition}
\end{align}
and
\begin{align}
C_u
&=
\frac{w_N}
{\bar R^{\mathrm{Nav}}}
\sum_{\substack{v=1\\v\neq u}}^{N}
\left[
2r_{\min}
+
v_{\max}\tau_{v,s}^{\max}
-
d_{uv}
\right]^+
\notag\\
&\quad+
\frac{
w_G\rho_{u,\mathrm{max}}\pi r_{\min}^{2}
}{
\bar R^{\mathrm{Ground}}}.
\label{eq:risk-constant-definition}
\end{align}
\end{proposition}

A detailed proof of Proposition~\ref{prop:risk-upper-bound}
is provided in Appendix~\ref{app:exposure-upper-bound}.

The upper bound in \eqref{eq:risk-upper-bound} depends on the AoI
\(h_{u,v}^{(w)}(t)\) of the oldest frame in the coherence window.
Tracking this state requires maintaining the completion history
of the previous \(w\) video frames. However, the scheduler makes
a scheduling decision for every frame, and completing the current
frame directly updates
\(h_{u,v}^{(1)}(t)=h_{u,v}(t)\), the AoI of the most recently
completed frame. We therefore use the latest-frame AoI as the
video state in the scheduling objective.

Motivated by the coefficients of the risk upper bound, we define
the following tractable per-UAV surrogate operational risk:
\begin{equation}
\widehat R_u(t)
=
\alpha_{u,s}h_{u,s}(t)
+
\alpha_{u,v}h_{u,v}(t).
\label{eq:surrogate-risk}
\end{equation}
The per-UAV surrogate operational risk in
\eqref{eq:surrogate-risk} retains the relative status and video
coefficients induced by the original operational risk model while
depending only on the AoI states used in the scheduling analysis.
It is employed as a tractable scheduling objective.

\noindent\textbf{Context timescale.} The operational context underlying $\alpha_{u,s}$ and $\alpha_{u,v}$, including UAV location and population exposure, generally evolves more slowly than the per-slot AoI and transmission states. We therefore analyze each fixed-context realization separately: the risk coefficients are held constant while the scheduler responds in every slot to changes in information freshness and video-transmission progress. 
Section~\ref{sec:eval} evaluates this two-timescale operation along moving UAV trajectories and dynamic parameters.

\subsection{Optimization Problem}
\label{sec:optimization-problem}

We use the per-UAV surrogate operational risk in
\eqref{eq:surrogate-risk} as the scheduling objective. The
long-term expected surrogate risk under policy
\(\pi\in\Pi\) is defined as
\begin{equation}
\widehat J^\pi
=
\limsup_{T\to\infty}
\frac1T
\sum_{t=1}^{T}
\sum_{u=1}^{N}
\mathbb E^\pi
\left[
\widehat R_u(t)
\right].
\label{eq:surrogate_system_risk}
\end{equation}
Our goal is to design a low-complexity scheduling policy $\pi$ that attempts to minimize $\widehat J^\pi$ subject to network resource constraints.
\begin{subequations}\label{eq:surrogate_optimization}
\begin{empheq}[box=\fbox]{align}
\widehat J^*
=
\min_{\pi\in\Pi}
&\quad
\left\{
\limsup_{T\to\infty}
\frac1T
\sum_{t=1}^{T}
\sum_{u=1}^{N}
\mathbb E^\pi
\left[
\widehat R_u(t)
\right]
\right\}
\label{eq:surrogate_optimization_obj}
\\
\textnormal{s.t.}\qquad
&
\sum_{u=1}^{N}
\bigl(
a_{u,s}(t)+a_{u,v}(t)
\bigr)
\le K,
\qquad
\forall t.
\label{eq:surrogate_optimization_constraint}
\end{empheq}
\end{subequations}

The alignment of the proposed surrogate risk~\eqref{eq:surrogate-risk}, the resulting scheduling policy~\eqref{eq:surrogate_optimization}, and SORA's operational criticality~\cite{sora2024} will be illustrated in Fig.~\ref{fig:sail-gamma-violin}, where we show that our scheduling policy consistently serves UAVs with greater criticality.

For a given system setup
\((N,K,L,p_{u,s},p_{u,v},\alpha_{u,s},\alpha_{u,v})\),
let \(\widehat J^*\) denote the minimum long-term expected
surrogate risk.
Similarly, for any admissible scheduling policy
\(\eta\in\Pi\),
let \(\widehat J^\eta\) denote the corresponding expected
surrogate risk achieved by policy \(\eta\).
We define the optimality ratio of policy \(\eta\) as
$\psi^\eta
=
\widehat J^\eta/
\widehat J^*. 
$
A policy \(\eta\) is said to be
\(\psi^\eta\)-optimal if its long-term expected surrogate risk is within a multiplicative factor
\(\psi^\eta\geq 1\) of the optimum.
The optimality ratio provides a quantitative measure of the
performance gap between a low-complexity scheduling policy and
the optimal policy. 

\section{Proposed 5G Scheduling Policy}
\label{sec:policy}
In this section, we develop a low-complexity scheduling policy for minimizing the long-term expected surrogate risk~\eqref{eq:surrogate_optimization} and obtain its optimality ratio under fixed operational context and risk
coefficients.
In \S~\ref{sec:related}, we will briefly discuss related optimization problems, computational-heavy optimal solutions, and low-complexity sub-optimal solutions. Two key contributions of this paper are: (i) applying the rich AoI theory to a pressing, practical problem; (ii) proposing a carefully designed Lyapunov function that improves the optimality ratio by a factor of at least $1.5\times$ with respect to~\cite{zhao2024large}. 
Next, we derive a universal lower bound \(L_B^{\star} \leq \widehat J^* \) on the optimal long-term expected surrogate risk using a sample-path argument.
Then, we develop a Risk-Aware Max-Weight scheduling policy based on a Lyapunov function. Finally, we use Lyapunov optimization to establish an upper bound on its optimality ratio $\psi^\eta\leq\widehat J^\eta/L_B^{\star}.$

\subsection{Lower Bound}
To derive a universal lower bound on~\eqref{eq:surrogate_optimization}, we characterize the minimum achievable surrogate risk over all feasible long-term UAV packet delivery rates. 

\begin{theorem}
\label{thm:lower_bound}
Assume that
\(0<p_{u,s}\leq1\) and \(0<p_{u,v}\leq1\) for every UAV \(u\).
For any network setup 
\((N,K,L,p_{u,s},p_{u,v},\)
\\
\(\alpha_{u,s},\alpha_{u,v})\),
the long-term expected surrogate risk under any admissible
policy \(\pi\in\Pi\) satisfies \(\widehat J^\pi \ge L_B^{\star},\)
where
\begin{subequations}
\begin{align}
L_B^{\star}
=
\min_{\{q_{u,s},q_{u,v}\}}
&\quad
\sum_{u=1}^{N}
\left(
\frac{\alpha_{u,s}}{2q_{u,s}}
+
\frac{\alpha_{u,v}}{2q_{u,v}}
\right)
\notag\\
&\qquad
+
\frac12
\sum_{u=1}^{N}
\left(
\alpha_{u,s}
+
(2L-1)\alpha_{u,v}
\right)
\label{eq:LB_opt_obj}
\\
\textnormal{s.t.}\qquad
&
\sum_{u=1}^{N}
\left(
\frac{q_{u,s}}{p_{u,s}}
+
\frac{Lq_{u,v}}{p_{u,v}}
\right)
\le K,
\label{eq:LB_opt_constraint1}
\\
&
0<q_{u,s}\le p_{u,s} \mbox{ and }
0<q_{u,v}\le
\frac{p_{u,v}}{L}, \forall u
\label{eq:LB_opt_constraint2}
\end{align}
\label{eq:LB}
\end{subequations}
Here,
\(q_{u,s}\) and \(q_{u,v}\) denote the long-term status packet delivery
rate and video-frame completion rate of UAV \(u\), respectively.
\end{theorem}
\begin{proof}[Proof Sketch]
Using a sample path argument, we characterize the expected surrogate risk in terms of the status-packet delivery and video-frame completion rates. Combining the resulting risk expression with the resource constraint~\eqref{eq:LB_opt_constraint1} yields the optimization problem in~\eqref{eq:LB}, whose optimal value gives the lower bound. Detailed derivations are provided in Appendix~\ref{proof_lowerbound}.
\end{proof}

\subsection{Risk-Aware 5G Scheduling Policy}
We develop a Max-Weight scheduling policy that is designed to reduce the surrogate risk by minimizing the expected one-slot drift of a Lyapunov function at every time slot $t$. A main challenge is the stochastic delay between action (i.e., packet scheduling) and reward (i.e., risk reduction), which makes the Lyapunov functions from prior works, including~\cite{kadota2018,kadota2019throughput,8933047}, unsuitable for our setting. To overcome this challenge, the authors in~\cite{zhao2024large} proposed a complex Lyapunov function that includes not only AoI but also a throughput debt with respect to a chosen throughput target. In contrast, we propose a simpler, thoughtfully designed risk-aware Lyapunov function that relies only on AoI and improves the resulting Max-Weight optimality ratio by a factor of $1.5\times$ with respect to~\cite{zhao2024large}. 
In \S~\ref{sec:eval}, we show that while \cite{zhao2024large} minimizes AoI, our risk-aware scheduling policy achieves up to 83\% lower  99th-percentile per-UAV composite operational risk, and lower risk    than~\cite{zhao2024large} at every evaluated operating point.


\noindent\textbf{Lyapunov Function.}
Denote the system state at the beginning of slot \(t\) as
$
S(t)
=
\{
h_{u,s}(t),
h_{u,v}(t),
g_u(t),
\ell_u(t)
\}_{u=1}^{N}
$. 
For the video-frame flow, define the video risk potential as follows 
\begin{equation}
\Phi_{u,v}(t) = \frac{\ell_u(t)}{L}(h_{u,v}(t)-1) + g_u(t) + 1 .
\label{eq:video_linear_potential}
\end{equation}
This video risk potential grows with video AoI $h_{u,v}(t)$, UAV frame age $g_u(t)$, and the number of remaining packets $\ell_u(t)$. Notably, the expected video risk potential reduces with a deterministic one-slot delay following a video packet transmission. 
We define the Lyapunov function as
\begin{equation}
L(S(t))
=
\sum_{u=1}^{N}
\left[
V_{u,s}h_{u,s}(t)
+
V_{u,v}\Phi_{u,v}(t)
\right],
\label{eq:MW_lyapunov_linear}
\end{equation}
where \(V_{u,s}>0\) and \(V_{u,v}>0\) are design parameters.

\noindent\textbf{One-slot Lyapunov Drift.}
The one-slot conditional Lyapunov drift is given by
\begin{equation}
\Delta(S(t))
=
\mathbb E
\left[
L(S(t+1))-L(S(t))
\mid S(t)
\right].
\label{eq:MW_drift_linear}
\end{equation}
By substituting the evolution of \(\ell_u(t)\), \(g_u(t)\), \(h_{u,s}(t)\), and \(h_{u,v}(t)\) from
\eqref{eq:remaining_packet},
\eqref{eq:system_time},
\eqref{eq:status_aoi},
and
\eqref{eq:video_aoi}, respectively, into the drift expression in~\eqref{eq:MW_drift_linear}, we obtain the upper bound on \(\Delta(S(t))\). The resulting bound is expressed in~\eqref{eq:MW_drift_bound_linear}-\eqref{eq:MW_video_weight_linear}, with detailed derivations provided in Appendix~\ref{app:lyapunov_drift_linear}
\begin{align}
\Delta(S(t))
\le
B(t)
&-
\sum_{u=1}^{N}
W_{u,s}(t)
\mathbb E
\left[
a_{u,s}(t)
\mid S(t)
\right]
\notag\\
&-
\sum_{u=1}^{N}
W_{u,v}(t)
\mathbb E
\left[
a_{u,v}(t)
\mid S(t)
\right].
\label{eq:MW_drift_bound_linear}
\end{align}
The scheduling independent term is
\begin{equation}
B(t)
=
\sum_{u=1}^{N}
V_{u,s}
+
\sum_{u=1}^{N}
V_{u,v}
\left(
1+\frac{\ell_u(t)}{L}
\right).
\label{eq:MW_B_term_linear}
\end{equation}
The status and video scheduling weights are
\begin{align}
W_{u,s}(t)
&=
V_{u,s}p_{u,s}h_{u,s}(t),
\label{eq:MW_status_weight_linear}
\\
W_{u,v}(t)
&=
\frac{V_{u,v}p_{u,v}}{L}h_{u,v}(t).
\label{eq:MW_video_weight_linear}
\end{align}

\noindent\textbf{Max-Weight Policy.}
According to
\eqref{eq:MW_drift_bound_linear},
minimizing the upper bound on the Lyapunov drift is
equivalent to maximizing the total scheduling weight.
Hence, in every slot \(t\), the Max-Weight policy solves
\begin{equation}
\begin{aligned}
\max_{\{a_{u,s}(t),a_{u,v}(t)\}}
\quad
&
\sum_{u=1}^{N}
a_{u,s}(t)W_{u,s}(t)
+
\sum_{u=1}^{N}
a_{u,v}(t)W_{u,v}(t)
\\
\textnormal{s.t.}
\quad
&
\sum_{u=1}^{N}
\left(
a_{u,s}(t)+a_{u,v}(t)
\right)
\le K,
\\
&
a_{u,s}(t),a_{u,v}(t)\in\{0,1\},
\quad \forall u.
\end{aligned}
\label{eq:MW_policy_linear}
\end{equation}
Equivalently, the scheduler selects the \(K\) streams with the
largest scheduling weights.

\noindent\textbf{Optimality Ratio for Max Weight Policy.}
We now establish a performance guarantee for the proposed
Max-Weight policy by comparing it with the lower bound derived
in Theorem~\ref{thm:lower_bound}.
\begin{theorem}
\label{thm:MW_linear_ratio}
Let
\(\{q_{u,s}^{\mathrm{LB}},q_{u,v}^{\mathrm{LB}}\}_{u=1}^{N}\)
be an optimal solution of
\eqref{eq:LB},
and choose the Lyapunov coefficients as
\begin{equation}
V_{u,s}
=
\frac{\alpha_{u,s}}{q_{u,s}^{\mathrm{LB}}},
\qquad
V_{u,v}
=
\frac{\alpha_{u,v}}{q_{u,v}^{\mathrm{LB}}}.
\label{eq:MW_linear_coefficients}
\end{equation}
Then the proposed Max-Weight policy satisfies
\begin{equation}
\widehat J^{MW}
\le
L_B^\star
+
\sum_{u=1}^{N}
\frac{\alpha_{u,s}}{2q_{u,s}^{\mathrm{LB}}}
+
\frac32
\sum_{u=1}^{N}
\frac{\alpha_{u,v}}{q_{u,v}^{\mathrm{LB}}},
\label{eq:MW_linear_upper_bound}
\end{equation}
where \(L_B^\star\) denotes the optimal value of
\eqref{eq:LB}.
Consequently,
\begin{equation}
\psi^{MW}
=
\frac{\widehat J^{MW}}{\widehat J^*}
\le
1+
\frac{
\displaystyle
\sum_{u=1}^{N}
\frac{\alpha_{u,s}}{2q_{u,s}^{\mathrm{LB}}}
+
\frac32
\sum_{u=1}^{N}
\frac{\alpha_{u,v}}{q_{u,v}^{\mathrm{LB}}}
}
{L_B^\star}.
\label{eq:MW_linear_ratio_bound}
\end{equation}
\end{theorem}

\begin{proof}
First, we compare the Max-Weight policy with a randomized
policy whose scheduling rates are chosen according
to the optimal solution of the lower-bound problem
\eqref{eq:LB}. By the Max-Weight
property, the one-slot Lyapunov drift under Max-Weight is no
larger than the drift under this comparison policy. Then,
choosing the Lyapunov coefficients as in
\eqref{eq:MW_linear_coefficients} makes the negative drift
terms coincide with the surrogate risk. Summing the resulting
drift inequality over time and applying telescoping gives
\eqref{eq:MW_linear_upper_bound}. Combining this upper bound
with \(L_B^\star\le \widehat J^*\) yields
\eqref{eq:MW_linear_ratio_bound}. Detailed derivations are
provided in Appendix~\ref{app:MW_linear_ratio_proof}.
\end{proof}

The performance guarantee applies to the fixed-context
surrogate objective \(\widehat J\), not directly to the
composite operational-risk metric \(R_u(t)\).
Section~\ref{sec:eval} evaluates their empirical alignment.

Theorem~\ref{thm:MW_linear_ratio} provides a parameterized
performance guarantee that explicitly depends on the optimal
value of the lower-bound problem. By using the inequality
\begin{equation}
 L_B^\star
\ge
\frac12
\sum_{u=1}^{N}
\left(
\frac{\alpha_{u,s}}{q_{u,s}^{\mathrm{LB}}}
+
\frac{\alpha_{u,v}}{q_{u,v}^{\mathrm{LB}}}
\right),   
\end{equation}
the ratio in
\eqref{eq:MW_linear_ratio_bound}
can be further bounded as
\begin{equation}
 \psi^{MW}
\le
1+
\frac{
\displaystyle
\sum_{u=1}^{N}
\frac{\alpha_{u,s}}{2q_{u,s}^{\mathrm{LB}}}
+
\frac32
\sum_{u=1}^{N}
\frac{\alpha_{u,v}}{q_{u,v}^{\mathrm{LB}}}
}{
\displaystyle
\frac12
\sum_{u=1}^{N}
\left(
\frac{\alpha_{u,s}}{q_{u,s}^{\mathrm{LB}}}
+
\frac{\alpha_{u,v}}{q_{u,v}^{\mathrm{LB}}}
\right)
}
\le
4.   
\end{equation}
Therefore, the proposed Max-Weight policy admits a
worst-case optimality ratio no larger than 4. Since this
bound is obtained by further relaxing the denominator of
\eqref{eq:MW_linear_ratio_bound}, it is conservative, and the
parameterized bound in Theorem~\ref{thm:MW_linear_ratio}
provides a tighter characterization whenever the value of
\(L_B^\star\) is available.

\section{Evaluation}
\label{sec:eval}

We evaluate RA against two conventional 5G schedulers, namely Round-Robin (RR) and proportional fair (PF), and against the AoI-minimizing policy developed in~\cite{zhao2024large} that shares RA's Max-Weight structure and update model but does not use operational context. 
RA uses the Max-Weight policy of \S~\ref{sec:policy} with Lyapunov coefficients $V_{u,s}$ and $V_{u,v}$ from~\eqref{eq:MW_linear_coefficients} and surrogate-risk coefficients $\alpha_{u,s}$ and $\alpha_{u,v}$ from~\eqref{eq:alpha-status-definition} and~\eqref{eq:alpha-video-definition}, respectively, that are computed based on the operational context, e.g., UAV location and population exposure. 

As discussed in \S~\ref{sec:intro}, RA obtains operational context from 5G core functions, such as the LMF and NWDAF, and from external sources, such as geographic databases and UAV reports. In general, operational context and the resulting Lyapunov coefficients $V_{u,s}$ and $V_{u,v}$ evolve on a slower timescale (seconds/minutes) than RA's status and video scheduling weights $W_{u,s}(t)$ and $W_{u,v}(t)$ from~\eqref{eq:MW_status_weight_linear} and~\eqref{eq:MW_video_weight_linear}, respectively, that may change in every time slot ($\leq 1$ ms). 

The analysis in Secs.~\ref{sec:stale-to-risk} and~\ref{sec:policy} considers a fixed operational context to establish the lower bound and performance guarantee for each context realization. In the evaluation, we allow the operational context to evolve along each UAV trajectory and assume that the gNB observes accurate and timely context information which is used to (re)compute RA's scheduling weights $W_{u,s}(t)$ and $W_{u,v}(t)$ in every time slot. 
%
The evaluation attempts to answer three questions:



\noindent\textbf{Q1: Does RA reduce risk exposure?}
In \S~\ref{sec:eval-risk}, we show that RA 
reduces 99th-percentile per-UAV risk $R_u(t)$ by up to $96\%$ over PF, $86\%$ over RR, and $83\%$ over AoI. 

\noindent\textbf{Q2: Does RA prioritize operationally critical UAVs?}
  In \S~\ref{sec:eval-sora}, we label operations offline with SORA's
  SAIL and show that, under contention,RA schedules a median $2.4\times$ more service per critical UAV
        than per regular UAV, while RR, PF, and AoI
  remain class-blind.


\noindent\textbf{Q3: How much operational capacity does RA enable?}
In \S~\ref{sec:eval-capacity}, we show that, within the same risk envelope, RA supports up to $1.5\times$ as many UAVs as PF, $3\times$ as many as RR, and $2$–$3\times$ as many as AoI. 

\subsection{Evaluation Setup}
\label{sec:eval-methodology}

\begin{figure}[t]
    \centering
    \includegraphics[width=\columnwidth]{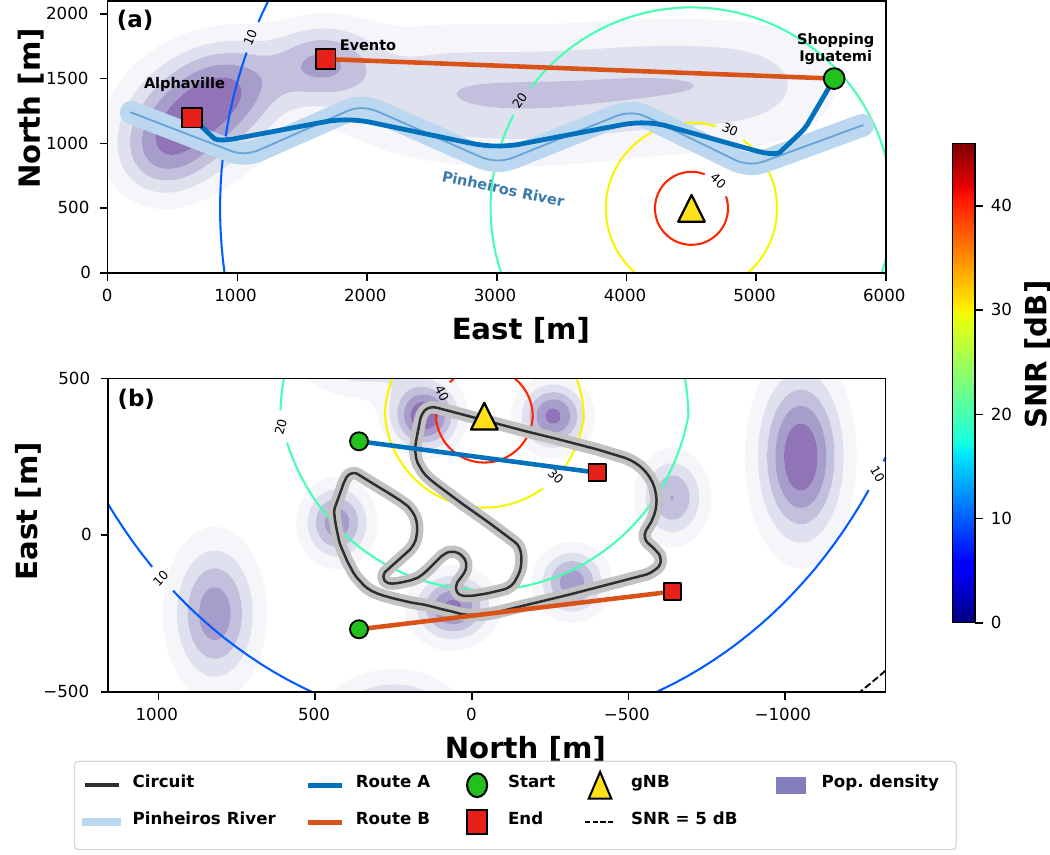}      
    \vspace{-0.4cm}
    \caption{Evaluation scenarios in S\~ao Paulo.
    Colors show the gNB-centered iso-SNR contours, the dashed contour
    marks the 5-dB outage boundary. Shaded regions represent
    population exposure. (a)~Delivery corridors near Shopping
    Iguatemi. (b)~Interlagos Formula~1.}
    \label{fig:scenarios}
    \vspace{-0.6cm}
\end{figure}

\noindent\textbf{Scenarios.}
We consider two scenarios derived from real operations
in S\~ao Paulo. Figure~\ref{fig:scenarios} overlays UAV
trajectories with iso-SNR contours and ground-exposure densities used in
the evaluation. The delivery scenario in Fig.~\ref{fig:scenarios}(a) contains two UAV corridors
with different combinations of coverage and population
exposure. The event scenario in Fig.~\ref{fig:scenarios}(b) follows the \(4.3\)-km
Interlagos circuit, where spectator hotspots create strongly
nonuniform exposure along the circuit. 
%
%
%
Appendix~\ref{app:scenarios-data} documents the public
infrastructure and operational data used to ground these
scenarios. 

\noindent\textbf{Network, traffic, and simulator.}
We implement a slot-level 5G NR uplink simulator in MATLAB to integrate the domains that determine communication-induced operational risk. On a common slot-level timeline, the simulator synchronizes UAV motion, position-dependent radio conditions, status and video traffic, packet delivery, information staleness, geographic population exposure, and scheduling decisions. UAV mobility is cross-checked against UAV Toolbox~\cite{toolboxuav}, while radio propagation is driven by 5G Toolbox~\cite{toolbox5g} models. This integration allows the operational consequences of each network decision to evolve with both the aircraft trajectory and the communication state. It also enables controlled comparisons in which all policies experience identical mobility, channel, traffic, and transmission-outcome realizations.

Each slot lasts \(1\)~ms, and the gNB provides at most
\(K\in\{1,2,4,8,16,32\}\) transmission opportunities to
$N\in\{5,10,20,$ $40,60,80,100,120\}$ concurrent UAVs,
yielding 48 operating points. For each of the four scheduling policies, we simulate every operating point for at least an $900$~s interval, during which all \(N\) UAVs simultaneously contend for the uplink channel. To ensure that different scheduling policies experience the same trajectories, channel states, and transmission outcomes, we use the same seed value.

UAVs fly at \(15\)~m/s from the start to the end of the corridors depicted in Fig.~\ref{fig:scenarios}. Each UAV carries the C2 status update and video streams. Status updates are generated with a frequency of 10~Hz and require a single successful packet transmission to be delivered to the gNB. Video frames are generated according to the realistic UAV traffic generator from \cite{aviator2021} with 
\(30\)~fps and a $15$-frame Group of Pictures (GOP) structure with one \(30\)-packet I-frame followed by fourteen \(3\)-packet P-frames. It follows that refreshing video may consume multiple transmission opportunities. While the number of packets per video frame changes over time, both RA and AoI use a fixed $L=10$ in the experiments below.

UAV dynamics are simulated using the UAV toolbox~\cite{toolboxuav}. 5G wireless channels are simulated using the 5G toolbox~\cite{toolbox5g} at \(3.5\)~GHz with a \(20\)-MHz bandwidth and an urban-macro channel model. UAV links with SNR below \(5\)~dB are in outage, i.e., $p_{u,s}=p_{u,v}=0$. Otherwise, scheduled packets succeed with probability \(p_{u,s}=p_{u,v}=0.8\). Additional toolbox parameters are provided in Appendix~\ref{app:toolbox}. 




\noindent\textbf{Simulator fidelity.} OAI provides the protocol-stack fidelity needed to demonstrate the scheduling blind spot and verify the resource-allocation behavior central to this study. MATLAB provides the cross-domain integration, experimental control, and scale needed to evaluate how that behavior affects operational risk across moving aircraft, geographic exposure conditions, traffic types, fleet sizes, and resource budgets. Because the complete \(K\times N\) sweep is impractical to reproduce with the full 5G stack, we validate the MATLAB abstraction against the OAI implementation used in Fig.~\ref{fig:blindspot} under matched saturated-load experiments using proportional-fair scheduling.

Specifically, we test symmetric channels, in which all UEs see
  the same radio conditions, with \(N\in\{2,3,4\}\), and
  asymmetric channels with \(N\in\{2,4\}\), in which half of the
  UEs are attenuated by about \(8\,\mathrm{dB}\) (MCS 20 versus
  28). Under symmetric channels, OAI and the simulator both give
  each UE the expected \(1/N\) share of uplink grants. Under
  asymmetric channels, PF equalizes time rather than bytes:
  degraded UEs deliver fewer bytes, but receive nearly the same
  grant share in both systems. Across all configurations, the
  maximum total-variation distance between their grant-share
  distributions is \(0.021\), about two percentage points. Thus,
  the simulator preserves the uplink allocation behavior that RA
  modifies.

\subsection{Does RA Reduce Risk Exposure?}
  \label{sec:eval-risk}
We first examine whether minimizing information age also minimizes the operational consequences of stale information. We compare the policies using two metrics: (i) 99th-percentile risk exposure ($p_{99,u}[\bar R_u]$), where $\bar R_u$ captures the time-average of the operational risk $R_u(t)$ for a particular UAV; and (ii) Fleet-Mean Age of Information ($\bar h$), where $\bar h=(\bar h_s+\bar h_v)/2,$ with $\bar h_s$ representing the C2 status update AoI $h_{u,s}(t)$ averaged over time and across all UAVs, and $\bar h_v$ representing the video frame AoI $h_{u,v}(t)$ averaged over time and across all UAVs. Figure~\ref{fig:risk-cost} places each policy on this 
  freshness--risk plane for three load regimes $\rho = N/K$, aggregating the operating points from both evaluation scenarios. The dashed line marks the communication-induced operational-risk envelope $p_{99,u}[\bar R_u]=1$. 

Figure~\ref{fig:risk-cost} shows that freshness and risk induce different policy rankings. The AoI policy achieves the lowest fleet-mean age in every load regime, yet it does not achieve the lowest tail risk. Even under light load, its risk exposure approaches the envelope. The reason is that AoI does not account for the operational consequence of staleness. It may therefore refresh an UAV whose information is old but operationally tolerable while deferring one whose growing uncertainty intersects traffic or population exposure. PF improves upon RR in both freshness and risk, showing that uniform service is inefficient. RA changes the allocation objective more fundamentally: it directs service toward the UAVs for which additional delay consumes the operational-risk margin most rapidly. This distinction becomes most consequential under moderate load and overload, where resources must be withheld from some UAVs.Across the complete 48-point evaluation in Fig.~\ref{fig:deep-ra},     RA attains the lowest $p_{99,u}[\bar R_u]$ at every operating point, with reductions of up to $96\%$ relative to PF, $86\%$ relative to
        RR, and $83\%$ relative to AoI.

  

  \begin{figure}[!t]
      \centering
      \includegraphics[width=\columnwidth]{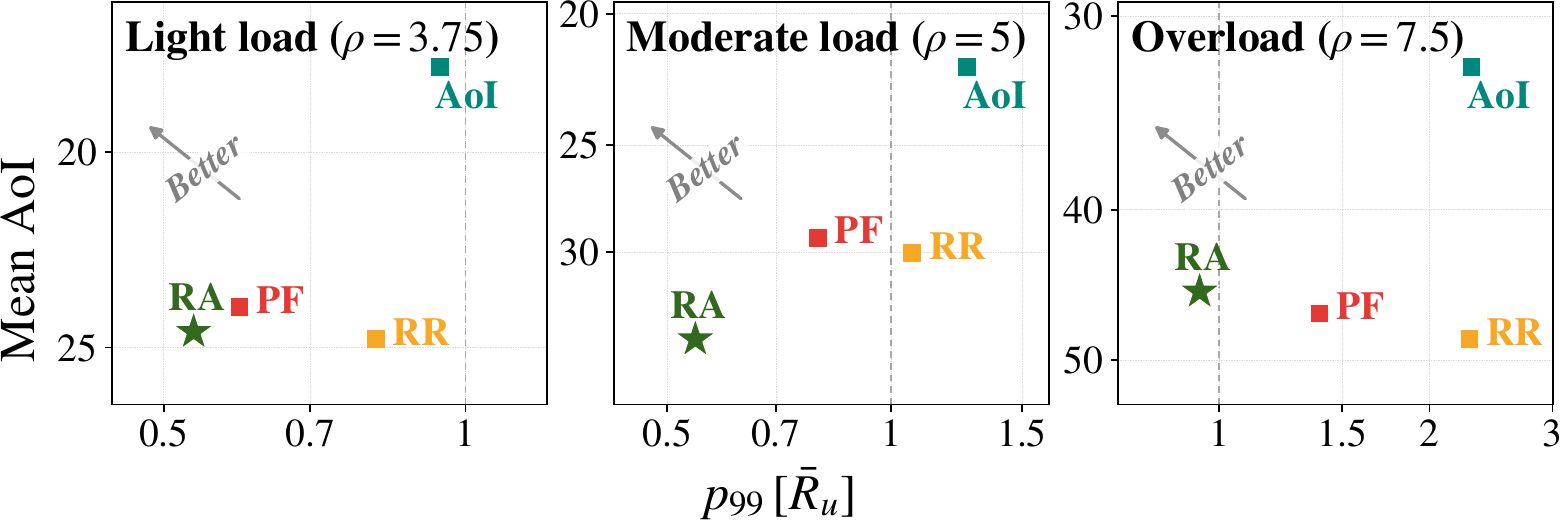}
      \vspace{-0.5cm}
     \caption{Each panel reports the fleet-mean AoI, $\bar h$, and the 99th percentile of time-average per-UAV operational risk, $p_{99,u}[\bar R_u]$, aggregated across both evaluation scenarios for operating points in the indicated load regime, $\rho = N/K$. Lower values are preferable on both axes. The dashed line denotes the communication-induced operational-risk envelope. A complete evaluation across all 48 operating points can be found in Fig.~\ref{fig:deep-ra} of Appendix~\ref{app:deep-ra}.}
      \label{fig:risk-cost}
      \vspace{-0.6cm}
  \end{figure}



  \noindent\textbf{Takeaway.} Minimizing average AoI is not sufficient for risk-sensitive UAV communication because equal staleness can have unequal operational consequences. Under contention, the relevant scheduling decision is not simply which information is oldest, but where an additional unit of delay consumes the operational-risk margin fastest. RA exploits this distinction by accepting a controlled increase in average AoI while substantially reducing tail risk. 

\subsection{Does RA Prioritize Operationally Critical UAVs?}
  \label{sec:eval-sora}


Section~\ref{sec:eval-risk} showed that RA reduces tail risk by directing freshness toward the UAV for which additional staleness is most consequential. We now test whether this behavior aligns with an external measure of operational criticality rather than only with RA’s own risk metric. We use SORA’s SAIL as this independent reference. SAIL classes I through VI (discussed in Appendix~\ref{app:sail-allocation}) rank operations according to the rigor required by their operational risk, with higher classes generally corresponding to more consequential operating conditions. We assign SAIL offline and do not expose it to RA or any baseline. 
In each slot $t$, we partition the UAVs according to whether they are operating under SAIL-II or higher-SAIL conditions. For UAV $u$, let $\overline{\theta}_{\mathrm{u,SAIL\text{-}II}}$ denote the fraction of slots in which UAV $u$ is scheduled among the slots in which it is classified as SAIL-II. We then define $\overline{\theta}_{\mathrm{SAIL\text{-}II}}$ as the mean of $\overline{\theta}_{\mathrm{u,SAIL\text{-}II}}$ across UAVs. We proceed analogously to define $\overline{\theta}_{\mathrm{u,Higher\text{-}SAIL}}$ and $\overline{\theta}_{\mathrm{Higher\text{-}SAIL}}$. 
%
%
  %
We then define the scheduling-alignment ratio $\Gamma_{\mathrm{sched}}
  = \overline{\theta}_{\mathrm{Higher\text{-}SAIL}}
  / \overline{\theta}_{\mathrm{SAIL\text{-}II}}$. A value of $\Gamma_{\mathrm{sched}}=1$ denotes class-blind scheduling, and $\Gamma_{\mathrm{sched}}>1$ indicates that the scheduler directs more service per UAV toward higher-SAIL operations. 
  

  Figure 5 reports $\Gamma_{\mathrm{sched}}$ for the 42 operating points with $N>K$, aggregated across both evaluation scenarios. At these points, the available transmission opportunities are insufficient to serve every UAV in each slot, so the policies must decide which UAV to defer. We exclude the remaining six points because $N\leq K$ allows every UAV to be scheduled simultaneously. Under those conditions, differences in $\Gamma_{\mathrm{sched}}$ primarily reflect traffic availability rather than prioritization.

  \begin{figure}[t]
  \centering
  \includegraphics[width=\columnwidth]{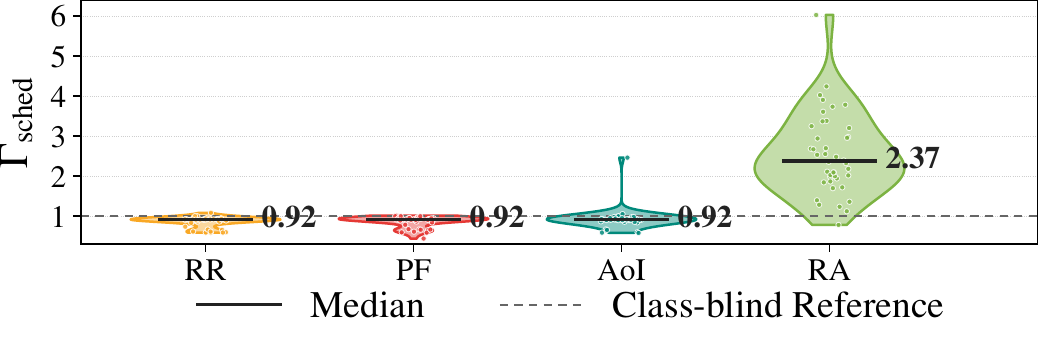}
  \vspace{-0.5cm}
  \caption{For each of the 42 contended operating points with $N>K$, aggregated across both evaluation scenarios, $\Gamma_{\mathrm{sched}}$ compares the mean per-UAV scheduling probability under higher-SAIL conditions with that under SAIL-II conditions. Each point represents one operating point, bars denote medians, and the dashed line at $\Gamma_{\mathrm{sched}}=1$ denotes class-blind allocation.}
  \label{fig:sail-gamma-violin}
  \vspace{-0.6cm}
  \end{figure}



Figure~\ref{fig:sail-gamma-violin} shows that the conventional and age-based baselines remain close to class-blind allocation. RR, PF, and AoI each have a median $\Gamma_{\mathrm{sched}}=0.92$, near the class-blind reference of one. Because these schedulers do not observe SAIL, we interpret the modest departure from one as scenario-induced variation rather than evidence of systematic SAIL class preference. RA exhibits a clearly different allocation pattern. Its median $\Gamma_{\mathrm{sched}}=2.37$ indicates that, under contention, a UAV operating above SAIL-II receives more than twice the scheduling probability of a UAV operating at SAIL-II. The ratio exceeds three at most operating points with $K\leq 2$ and reaches six at $K=1$, $N=100$. This aggregate difference is not explained solely by higher-SAIL UAVs having older information or different channel conditions. In the controlled scenario in Appendix~\ref{app:sail-allocation}, where candidate UAVs have comparable AoI and radio state, RA selects the higher-SAIL operation in 64\% of matched decisions (median over operating points), against 50\% for the baselines. Moreover, Fig.~\ref{fig:sail-gamma-violin} shows that the allocation ratio generally increases with contention, indicating that RA’s operational context becomes most influential precisely when the scheduler must make more selective deferral decisions. RA produces this alignment without observing SAIL, which is used only as an offline evaluation label.


\noindent\textbf{Takeaway.} RA’s risk reduction is explained by where it places scarce transmission opportunities, not by uniformly improving freshness. Under contention, RA assigns more than the per-UAV scheduling probability to higher-SAIL operations, and this preference persists when comparing UAVs with similar AoI and radio conditions. Because SAIL is neither an input to RA nor part of its optimization objective, this alignment provides an independent check of the abstractions leading to the scheduler. Although it does not validate each modeling assumption individually, it indicates that the assumptions used to define the composite operational-risk metric, $R_u(t)$, derive its tractable surrogate, $\widehat R_u(t)$, formulate the scheduling optimization problem in~\eqref{eq:surrogate_optimization}, and derive RA collectively preserve the operational distinctions relevant to scheduling. The resulting network decisions are therefore consistent with an independent operation-level risk classification. 

\subsection{How Much Operational Capacity Does RA Enable?}
\label{sec:eval-capacity}

We next translate RA’s reduction in operational-risk into communication-constrained operational capacity. For a policy $\pi$ scheduling up to $K$ UAV transmissions in each slot, we define the operational capacity $N_{\pi}^{\star}(K)$ as the largest evaluated fleet for which every UAV remains within the communication-induced operational-risk envelope, namely $N_{\pi}^{\star}(K)=\max\{N:\max_u \bar R_u\le1 \text{ at }(K,N)\}$.
We use the worst-UAV criterion because admitting a fleet based on its average risk could conceal an individual UAV that persistently exceeds the envelope. 

Figure~\ref{fig:max-risk} reports \(\max_u \bar R_u\)  for every $(N,K)$ pair, aggregated across both evaluation scenarios. Values at or below one identify the admissible region for each policy.



RA produces the largest admissible region when the channel budget is constrained but still sufficient for prioritization to affect the risk boundary. At $K=8$, RA supports 60 simultaneous UAVs, compared with 40 under PF and 20 under RR and AoI. At $K=16$, RA supports the full evaluated fleet of 120 UAVs, compared with 80 under PF and 60 under RR and AoI. Thus, within the evaluated range, RA supports up to $1.5\times$ as many UAVs as PF, $3\times$ as many as RR, and $2$ to $3\times$ as many as AoI while enforcing the same communication-induced operational-risk envelope. These gains arise without increasing the channel budget. RA instead changes which UAVs and information streams receive the available transmission opportunities.  A complete evaluation across all 48 operating points can be found in Appendix~\ref{app:deep-ra}.


Section~\ref{sec:eval-risk} used $p_{99,u}[\bar R_u]$ to compare persistent tail exposure, whereas operational capacity applies the stricter condition \(\max_u \bar R_u\leq 1\) to ensure that no individual UAV exceeds the envelope. This stricter criterion does not reduce the measured capacity on the evaluated grid. Replacing $\max_u \bar R_u$ with $p_{99,u}[\bar R_u]$ yields the same admissible $(K,N)$ region for every policy. Within that region, the worst-UAV risk exceeds the 99th percentile by only 1\% on median and by at most 13\%. The capacity gains are therefore not obtained by improving most UAVs while leaving a single aircraft substantially more exposed. Over the evaluated operating points, RA enlarges the admissible fleet under both a fleet-tail criterion and an individual-UAV criterion.


\begin{figure}
  \centering
  \includegraphics[width=\columnwidth]{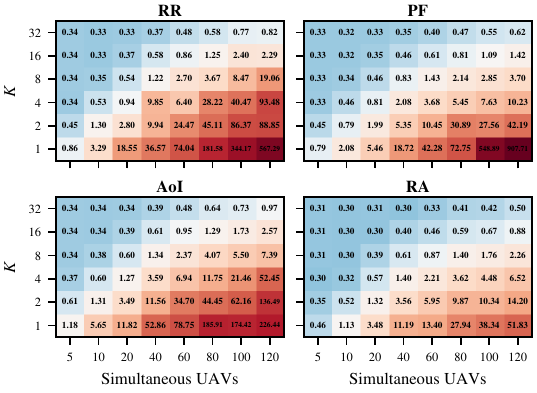}

  \vspace{-0.2cm}
  \caption{
  Each cell reports \(\max_u \bar R_u\) for one \((K,N)\) operating
  point. Values at or below one satisfy the communication-induced operational-risk envelope and admit the corresponding evaluated fleet. Values above one violate the envelope.
  }
  \label{fig:max-risk}
  \vspace{-0.5cm}
  \end{figure}

Figure~\ref{fig:max-risk} also separates failures of resource allocation from resource insufficiency. When a baseline exceeds the risk envelope but RA remains admissible, the available radio resources are sufficient, but the baseline allocates them inconsistently with operational consequence. When every policy exceeds the envelope, the channel budget itself may be insufficient for the offered fleet. In this regime, the appropriate response may be to increase communication capacity (e.g., by expanding the bandwidth of the UAV network slice) or reduce traffic demand through UAV departure spacing, speed adjustment, or rerouting. Communication-constrained operational capacity therefore provides an interface between millisecond-scale scheduling and longer-timescale traffic-management decisions. Appendix~\ref{app:operational-interface} discuss how it could be used in air traffic control and air traffic flow management.

  \noindent\textbf{Takeaway.} Operational capacity depends not only on the amount of available radio resources, but also on whether those resources are allocated according to the operational consequences of delay. Within the allocation-sensitive regime, RA supports $1.5$ to $3\times$ larger evaluated fleets than the baselines under the same communication-induced operational-risk envelope. This gain requires no additional spectrum or transmission opportunities. It follows from redirecting service toward the aircraft and streams for which deferral is most consequential. However, risk-aware allocation cannot overcome an aggregate service deficit. Once every policy exceeds the envelope, maintaining admissibility requires additional communication capacity or traffic-management actions that reduce simultaneous demand.

\section{Related Work}
\label{sec:related}

\noindent\textbf{Cellular networks and operational risk.}
Prior work on cellular-connected UAVs characterizes aerial
propagation, interference, mobility, handovers, and three-dimensional
coverage, and optimizes association, trajectory planning, and
radio-resource management
\cite{zeng2019cellular,liu2026lowaltitude,liu5gsky}.
Connectivity-aware planners further use radio maps or predicted
outages to adapt the UAV trajectory
\cite{5g_ground_risk,ouahouah2024}.
3GPP QoS, slicing, and reservation mechanisms differentiate traffic
through rates, latency, reliability, or service classes
\cite{foukas2017orion,zhang2026smec}.
These approaches improve connectivity, but do not determine which
aircraft is most costly to delay when they
contend for the same uplink.

A parallel body of UAS safety work models ground exposure from
population and impact footprints and air risk from navigation
uncertainty and separation
\cite{asghari2025,ground_risk_dynamic,kim2022,rtcaDO365}.
SORA and related assurance frameworks further make C3 performance
part of the operational safety case
\cite{sora2024,jarus_sora25_annexf,ed267}.
These methods assess or constrain operations, but generally treat
communication performance as an input. RA instead turns these
operational quantities into a network-computable risk-driven metric used
online for radio-resource allocation.

\noindent\textbf{Freshness- and context-aware scheduling.}
Age of Information motivates scheduling policies that minimize
average, weighted, or peak staleness under wireless constraints
\cite{kaul2012real,kadota2018,kadota2019throughput,
sun2017,huang2015,zhao2025}, while recent work handles updates 
requiring multiple transmission opportunities
\cite{zhao2024large}. 
Closest to RA, Urgency of Information (UoI) extends AoI-type 
timeliness by combining a time-varying context weight with
estimation-error cost and uses this urgency for multi-terminal
status-update scheduling~\cite{Zheng2020uoi}. 
UoI therefore establishes that context can matter beyond age alone.
However, it assumes the context weight as an input to the model. 
The paper identifies choosing practical context weights as
an open problem. 
update supersedes older status information.

RA addresses a different systems question: how a 5G scheduler can
derive such prioritization from an aviation operational-safety model.
We map network-observable freshness and aircraft context to
normalized navigation, ground, and video-coherence exposure
\cite{aviator2021,baltaci2022,mcguire2017}, and jointly schedule
single-slot status updates and multi-slot video objects. The resulting
risk envelope further translates this online allocation into
communication-constrained operational capacity.

\section{Conclusion}
\label{sec:conc}
Low-altitude aviation exposes a limitation of conventional cellular scheduling: network-centric metrics do not capture the operational consequences of delaying information from different aircraft. We introduce a network-to-operations abstraction that translates information freshness and operational context into aircraft-specific operational risk, enabling the network to evaluate not only how efficiently a transmission can be served, but also the consequence of making it wait. We instantiate this abstraction through RA scheduling, which prioritizes the aircraft and information streams for which additional delay is most consequential while accounting for heterogeneous transmission costs. Although RA is derived from an analytically tractable model, it is evaluated under moving UAV trajectories, dynamic operational context, geographically varying population exposure, and realistic UAV video traffic. Across 48 operating points and real S\~ao Paulo operating scenarios, RA substantially reduces communication-induced operational risk and increases the number of UAVs supportable within the same risk envelope. More broadly, this work illustrates how future communication systems supporting low-altitude aviation can incorporate operational consequence directly into radio-resource allocation.


\bibliographystyle{ACM-Reference-Format}
\bibliography{reference}

@inproceedings{kaul2012real,
  author    = {Sanjit Kaul and Roy D. Yates and Marco Gruteser},
  title     = {Real-Time Status: How Often Should One Update?},
  booktitle = {Proc. IEEE INFOCOM'12},
  year      = {2012},
  pages     = {2731--2735}
}

@article{kadota2018,
  author  = {Igor Kadota and Abhishek Sinha and
             Elif Uysal-Biyikoglu and Rahul Singh and Eytan Modiano},
  title   = {Scheduling Policies for Minimizing Age of Information
             in Broadcast Wireless Networks},
  journal = {IEEE/ACM Transactions on Networking},
  volume  = {26},
  number  = {6},
  year    = {2018},
  pages   = {2637--2650}
}

@article{kadota2019throughput,
  author  = {Igor Kadota and Abhishek Sinha and Eytan Modiano},
  title   = {Scheduling Algorithms for Optimizing Age of Information
             in Wireless Networks with Throughput Constraints},
  journal = {IEEE/ACM Transactions on Networking},
  volume  = {27},
  number  = {4},
  year    = {2019},
  pages   = {1359--1372}
}

@inproceedings{zhao2025,
  author    = {Zhuoyi Zhao and Igor Kadota},
  title     = {Optimizing Age of Information without Knowing the Age of Information},
  booktitle = {Proc. IEEE INFOCOM'25},
  year      = {2025},
  pages     = {1--10}
}

@inproceedings{zhao2024large,
  author    = {Zhuoyi Zhao and Vishrant Tripathi and Igor Kadota},
  title     = {Optimizing Age of Information in Networks with Large and Small Updates},
  booktitle = {Proc. IEEE WiOpt'25},
  year      = {2025},
  pages     = {1--8}
}

@article{sun2017,
  author  = {Yin Sun and Elif Uysal-Biyikoglu and Roy D. Yates and C. Emre Koksal and Ness B. Shroff},
  title   = {Update or Wait: How to Keep Your Data Fresh},
  journal = {IEEE Trans. Inf. Theory},
  volume  = {63},
  number  = {11},
  pages   = {7492--7508},
  month   = nov,
  year    = {2017}
}

@inproceedings{huang2015,
  author    = {Longbo Huang and Eytan Modiano},
  title     = {Optimizing Age-of-Information in a Multi-Class Queueing System},
  booktitle = {Proc. IEEE ISIT'15},
  year      = {2015},
  pages     = {1681--1685}
}

@inproceedings{baltaci2022,
  author    = {Ayg{\"u}n Baltaci and Hendrik Cech and Nitinder Mohan and Fabien Geyer and Vaibhav Bajpai and J{\"o}rg Ott and Dominic Schupke},
  title     = {Analyzing Real-Time Video Delivery over Cellular Networks for Remote Piloting Aerial Vehicles},
  booktitle = {Proc. ACM IMC'22},
  year      = {2022},
  pages     = {98--112}
}

@inproceedings{aviator2021,
  author    = {Ayg{\"u}n Baltaci and Markus Kl{\"u}gel and Fabien Geyer and Svetoslav Duhovnikov and Vaibhav Bajpai and J{\"o}rg Ott and Dominic Schupke},
  title     = {Experimental {UAV} Data Traffic Modeling and Network Performance Analysis},
  booktitle = {Proc. IEEE INFOCOM'21},
  year      = {2021},
  pages     = {1--10}
}

@article{asghari2025,
  author    = {Omid Asghari and Naghmeh Ivaki and Henrique Madeira},
  title     = {{UAV} Operations Safety Assessment: A Systematic Literature Review},
  journal   = {ACM Comput. Surv.},
  volume    = {57},
  number    = {10},
  articleno = {266},
  numpages  = {37},
  year      = {2025}
}

@article{alipour2024datmo,
  author  = {MReza Alipour Sormoli and Mehrdad Dianati and Sajjad Mozaffari and Roger Woodman},
  title   = {Optical Flow Based Detection and Tracking of Moving Objects for Autonomous Vehicles},
  journal = {IEEE Trans. Intell. Transp. Syst.},
  volume  = {25},
  number  = {9},
  pages   = {12578--12590},
  year    = {2024}
}

@inproceedings{ouahouah2024,
  author    = {Sihem Ouahouah and Mohammed Lahouari Harchaoui and Oussama Bekkouche and Miloud Bagaa and Riku J{\"a}ntti},
  title     = {{5G}-Based Autonomous Ground Risk Mitigation for {UAVs}},
  booktitle = {Proc. IEEE ICC'25},
  year      = {2025},
  pages     = {1359--1364}
}

@inproceedings{5g_ground_risk,
  author    = {Mohammed Lahouari Harchaoui and Sihem Ouahouah and Oussama Bekkouche and Miloud Bagaa and Abir Derouiche},
  title     = {{5G}-Based Ground Risk Mitigation for {UAVs}: A Deep Reinforcement Learning Approach},
  booktitle = {Proc. IEEE GLOBECOM'24},
  pages     = {1539--1544},
  year      = {2024}
}

@techreport{sora2024,
  author       = {{JARUS WG-SRM}},
  title        = {{JARUS Guidelines on Specific Operations Risk Assessment (SORA), Edition 2.5}},
  institution  = {Joint Authorities for Rulemaking on Unmanned Systems (JARUS)},
  year         = {2024},
  number       = {JAR-DEL-SRM-SORA-MB-2.5}
}

@techreport{ed253,
  author       = {{EUROCAE WG-105}},
  title        = {{ED-253: Operational Services and Environment Definition (OSED) 
                   for Automation and Emergency Recovery for UAS}},
  institution  = {EUROCAE},
  year         = {2019}
}

@techreport{ed267,
  author       = {{EUROCAE WG-105}},
  title        = {{ED-267: Operational Services and Environment Definition (OSED) 
                   for Detect and Avoid in Very Low-Level Operations}},
  institution  = {EUROCAE},
  year         = {2020}
}

@techreport{ed272,
  author       = {{EUROCAE WG-105}},
  title        = {{ED-272: Minimum Aviation Systems Performance Standards for Remote 
                   Pilot Stations Conducting IFR Operations in Controlled Airspace}},
  institution  = {EUROCAE},
  year         = {2020},
}

@techreport{faaConops,
  author       = {{Federal Aviation Administration}},
  title        = {{Unmanned Aircraft System (UAS) Traffic Management (UTM) 
                   Concept of Operations, Version 2.0}},
  institution  = {FAA},
  year         = {2020}
}

@techreport{sesarConops,
  author       = {{SESAR Joint Undertaking}},
  title        = {{U-space Concept of Operations (ConOps), Fourth Edition}},
  institution  = {SESAR Joint Undertaking},
  year         = {2023}
}

@techreport{rtcaDO365,
  author       = {{RTCA SC-228}},
  title        = {{DO-365C: Minimum Operational Performance Standards (MOPS) 
                   for Detect and Avoid (DAA) Systems}},
  institution  = {RTCA},
  year         = {2022}
}

@inproceedings{kafetzis2025,
  author    = {Dimitrios Kafetzis and Nikos Fotiou and Savvas Argyropoulos and Jad Nasreddine and Iordanis Koutsopoulos},
  title     = {Video Quality Monitoring for Remote Autonomous Vehicle Control},
  booktitle = {Proc. IEEE MeditCom'25},
  pages     = {1--7},
  year      = {2025}
}

@article{mcguire2017,
  author  = {Kimberly McGuire and Guido {de Croon} and Christophe {De Wagter} and Karl Tuyls and Hilbert Kappen},
  title   = {Efficient Optical Flow and Stereo Vision for Velocity Estimation and Obstacle Avoidance on an Autonomous Pocket Drone},
  journal = {IEEE Robotics and Automation Letters},
  volume  = {2},
  number  = {2},
  pages   = {1070--1076},
  year    = {2017}
}

@article{kim2022,
  author  = {Younsil Kim and Joongwon Bae},
  title   = {Risk-Based {UAV} Corridor Capacity Analysis above a Populated Area},
  journal = {Drones},
  volume  = {6},
  number  = {9},
  pages   = {221},
  year    = {2022}
}

@techreport{3gpp_tr38857,
  author       = {{3GPP}},
  title        = {{Study on NR Positioning Support}},
  institution  = {3rd Generation Partnership Project (3GPP)},
  type         = {Technical Report (TR)},
  number       = {38.857},
  year         = {2021},
  note         = {Release 17, V17.0.0},
  url          = {https://www.3gpp.org/ftp/Specs/archive/38_series/38.857/}
}

@techreport{USDOT_AAM_National_Strategy_2025,
  title        = {Advanced Air Mobility National Strategy 2025},
  author       = {{U.S. Department of Transportation}},
  institution  = {U.S. Department of Transportation},
  year         = {2025}
}

@techreport{USDOT_AAM_Comprehensive_Plan_2025,
  title        = {Advanced Air Mobility Comprehensive Plan 2025},
  author       = {{U.S. DoT}},
  institution  = {U.S. Department of Transportation},
  year         = {2025}
}

@article{ground_risk_dynamic,
  author  = {Stefano Primatesta and Alessandro Rizzo and Anders la Cour-Harbo},
  title   = {Ground Risk Map for Unmanned Aircraft in Urban Environments},
  journal = {Journal of Intelligent and Robotic Systems},
  year    = {2020},
  volume   = {97},
  number   = {3--4},
  pages    = {489--509}
}

@misc{keeta_dubai,
  author       = {{KrASIA}},
  title        = {Meituan's Drone Delivery Debuts in {Dubai}, Eyeing a Future in Last-Mile Logistics},
  howpublished = {\url{https://kr-asia.com/meituans-drone-delivery-debuts-in-dubai-eyeing-a-future-in-last-mile-logistics}},
  year         = {2024},
  month        = dec
}

@misc{anac_speedbird,
  author       = {{sUAS News}},
  title        = {Delivery Drones Cleared for Flight over Populated Cities in {Brazil}},
  howpublished = {\url{https://www.suasnews.com/2026/03/delivery-drones-cleared-for-flight-over-populated-cities-in-brazil/}},
  year         = {2026},
  month        = mar
}

@techreport{dca351-6,
  author      = {{DECEA}},
  title       = {{Concepção Operacional UTM Nacional}},
  institution = {{Comando da Aeronáutica, Departamento de Controle do Espaço Aéreo}},
  type        = {{Diretriz do Comando da Aeronáutica}},
  number      = {{DCA 351-6}},
  address     = {Rio de Janeiro, Brasil},
  year        = {2022},
  url         = {https://publicacoes.decea.mil.br/publicacao/dca-351-6},
  urldate     = {2026-07-22}
}

@article{liu5gsky,
  author    = {Yanbing Liu and Jingqi Huang and Ziyu Li and Chunyi Peng},
  title     = {{5G} in the Sky: Uplink Throughput Measurement, Analysis, and Enhancement},
  journal   = {IEEE Transactions on Networking},
  volume    = {34},
  pages     = {941--956},
  year      = {2026},
  publisher = {IEEE}
}

@inproceedings{liu2026lowaltitude,
  author    = {Xinzhe Liu and Jianer Zhou and Xiaoyong Ni and Ke Luo and Zhenyu Li and Xiaofeng Tao and Weichao Li},
  title     = {Unveiling Low-Altitude {5G} Performance: Linking Key Influencing Factors with {UAV} Flight Parameters},
  booktitle = {Proceedings of the ACM SIGCOMM 2026 Conference},
  pages     = {317--333},
  year      = {2026},
  publisher = {Association for Computing Machinery},
  address   = {New York, NY, USA},
  location  = {Denver, CO, USA},
  isbn      = {979-8-4007-2467-1}
}

@misc{meituan,
  author       = {China Daily},
  title        = {Meituan First Nationwide Low-Altitude Operating License},
  howpublished = {\url{https://www.chinadaily.com.cn/a/202504/22/WS6807752fa3104d9fd3820e14.html}},
  year         = {2025},
  month        = apr
}

@techreport{ts23273,
  author = {{3GPP}}, title = {5G System (5GS) Location Services (LCS); Stage 2},
  type = {TS}, number = {23.273}, institution = {3rd Generation Partnership Project},
  note = {Release 19}, url = {https://www.3gpp.org/DynaReport/23273.htm},
  year = 2024}

@techreport{ts23288,
  author = {{3GPP}}, title = {Architecture enhancements for 5G System (5GS) to support network data analytics services},
  type = {TS}, number = {23.288}, institution = {3rd Generation Partnership Project},
  note = {Release 19}, url = {https://www.3gpp.org/DynaReport/23288.htm},
  year = {2025}}

@inproceedings{foukas2017orion,
  author    = {Xenofon Foukas and Mahesh K. Marina and Kimon Kontovasilis},
  title     = {{Orion}: {RAN} Slicing for a Flexible and Cost-Effective Multi-Service Mobile Network Architecture},
  booktitle = {Proc. ACM MobiCom'17},
  pages     = {127--140},
  year      = {2017}
}

@inproceedings{zhang2026smec,
  author    = {Xiao Zhang and Daehyeok Kim},
  title     = {Enabling {SLO-Aware} {5G} Multi-Access Edge Computing with {SMEC}},
  booktitle = {Proc. USENIX NSDI'26},
  pages     = {1759--1776},
  year      = {2026}
}

@article{8933047,
  author    = {Igor Kadota and Eytan Modiano},
  title     = {Minimizing the Age of Information in Wireless Networks with Stochastic Arrivals},
  journal   = {IEEE Transactions on Mobile Computing},
  volume    = {20},
  number    = {3},
  pages     = {1173--1185},
  year      = {2021},
  publisher = {IEEE}
}

@techreport{jarus_sora25_annexf,
  author      = {{Joint Authorities for Rulemaking on Unmanned Systems (JARUS)}},
  title       = {{JARUS Guidelines on SORA: Annex F -- Theoretical Basis for Ground Risk Classification and Mitigation}},
  institution = {Joint Authorities for Rulemaking on Unmanned Systems (JARUS)},
  number      = {JAR-DEL-SRM-SORA-F-2.5},
  year        = {2024},
  month       = may,
  note        = {Edition 2.5, Public Release, 13 May 2024},
  url         = {https://jarus-rpas.org/wp-content/uploads/2024/06/SORA-v2.5-Annex-F-Release.JAR_doc_29pdf.pdf}
}

@article{zeng2019cellular,
  author    = {Shuowen Zhang and Yong Zeng and Rui Zhang},
  title     = {Cellular-Enabled {UAV} Communication: A Connectivity-Constrained Trajectory Optimization Perspective},
  journal   = {IEEE Transactions on Communications},
  volume    = {67},
  number    = {3},
  pages     = {2580--2604},
  year      = {2019},
  publisher = {IEEE}
}

@article{Zheng2020uoi,
  author    = {Xi Zheng and Sheng Zhou and Zhisheng Niu},
  title     = {Urgency of Information for Context-Aware Timely Status Updates in Remote Control Systems},
  journal   = {IEEE Transactions on Wireless Communications},
  volume    = {19},
  number    = {11},
  pages     = {7237--7250},
  year      = {2020},
  publisher = {IEEE}
}

@misc{mundogeo26speedbird,
  author = {{MundoGEO}},
   title  = {Speedbird Aero and iFood Launch New Drone Delivery Route in
             S\~ao Paulo, Brazil},
   howpublished = {MundoGEO},
   month  = jun, year = {2026},
   url    = {https://mundogeo.com/en/2026/06/04/speedbird-aero-and-ifood-launch-new-drone-delivery-routein-sao-paulo-brazil/}
 }

@misc{ifood26drone,
  author = {{iFood Institucional}},
   title  = {Drone do iFood em S\~ao Paulo: empresa inicia entregas a\'ereas},
   howpublished = {iFood Institucional},
   year   = {2026},
   url    = {https://institucional.ifood.com.br/inovacao/drone-ifood-em-sao-paulo/}
 }

@misc{mundogeo25lagos,
  author = {{MundoGEO}},
   title  = {DECEA detectou drones irregulares e assegurou o fluxo a\'ereo
             no GP de S\~ao Paulo de F1},
   howpublished = {MundoGEO},
   month  = nov, year = {2025},
   url    = {https://mundogeo.com/2025/11/17/decea-detectou-drones-irregulares-e-assegurou-o-fluxo-aereo-no-gp-de-sao-paulo-de-f1/}
 }

@misc{decea24lagos,
  author = {{DECEA}},
   title  = {Opera\c{c}\~ao Lagos: DECEA controla 224 helic\'opteros e 82
             drones durante GP Brasil, em Interlagos (SP)},
   howpublished = {DECEA},
   year   = {2024},
   url    = {https://decea.mil.br/?i=midia-e-informacao&materia=operacao-lagos-decea-controla-224-helicopteros-e-82-drones-durante-gp-brasil-em-interlagos-sp&p=pg_noticia}
 }

@misc{anatel_mosaico_stations,
  author       = {{Agência Nacional de Telecomunicações (ANATEL)}},
  title        = {{Mosaico Spectrum-E: Consulta Pública de Emissões
                   Licenciadas}},
  year         = {2026},
  howpublished = {\url{https://sistemas.anatel.gov.br/se/public/view/b/licenciamento.php?view=licenciamento}},
  note         = {Accessed August 25, 2026}
}

@article{lin2024camp,
author = {Lin, Yu-Tai and Matson, N. Cameron and Sundaresan, Karthikeyan},
title = {Bringing Collaborative Positioning to Native 5G Systems for Enhanced 2D \& 3D Location Services},
year = {2024},
issue_date = {December 2024},
publisher = {Association for Computing Machinery},
address = {New York, NY, USA},
volume = {2},
number = {CoNEXT4},
url = {https://doi.org/10.1145/3696380},
doi = {10.1145/3696380},
month = nov,
articleno = {23},
numpages = {21}
}

@article{ahmad2023aerotraj,
author = {Ahmad, Fawad and Shin, Christina Suyong and Ghosh, Rajrup and D'Ambrosio, John and Chai, Eugene and Sundaresan, Karthikeyan and Govindan, Ramesh},
title = {AeroTraj: Trajectory Planning for Fast, and Accurate 3D Reconstruction Using a Drone-based LiDAR},
year = {2023},
issue_date = {September 2023},
publisher = {Association for Computing Machinery},
address = {New York, NY, USA},
volume = {7},
number = {3},
url = {https://doi.org/10.1145/3610911},
doi = {10.1145/3610911},
month = sep,
articleno = {83},
numpages = {28}
}

@techreport{uk_conops,
  author      = {{UK Civil Aviation Authority}},
  title       = {Concept of Operations: Future Air Traffic Management
                 and Air Navigation Service},
  institution = {UK Civil Aviation Authority},
  type        = {Civil Aviation Publication},
  number      = {CAP 3269},
  month       = jul,
  year        = {2026},
  url         = {https://www.caa.co.uk/publication/download/29467},
  urldate     = {2026-08-28}
}

@techreport{uk_roadmap,
  author      = {{UK Civil Aviation Authority}},
  title       = {Future of Flight: {BVLOS} Roadmap},
  institution = {UK Civil Aviation Authority},
  type        = {Civil Aviation Publication},
  number      = {CAP 3182},
  month       = oct,
  year        = {2025},
  url         = {https://www.caa.co.uk/publication/download/26211},
  urldate     = {2026-08-28}
}

@misc{toolboxuav,
  author = {{The MathWorks, Inc.}},
  title  = {{UAV Toolbox} Documentation},
  year   = {2026},
  url    = {https://www.mathworks.com/help/uav/index.html},
  note   = {Accessed September 1, 2026}
}

@misc{toolbox5g,
  author = {{The MathWorks, Inc.}},
  title  = {{5G Toolbox} Documentation},
  year   = {2026},
  url    = {https://www.mathworks.com/help/5g/index.html},
  note   = {Accessed September 1, 2026}
}

\appendix

\section{Normalization of Operational Risk Components}
\label{app:risk-normalization}

The navigation, ground, and video-coherence risk components
measure heterogeneous quantities: a separation deficit, an
exposed population, and a video-staleness violation,
respectively. To aggregate them, we normalize each component
using a conservative reference value constructed from the
corresponding declared operational limits. These reference
values define dimensionless risk scales.

\noindent\textbf{Declared Operational Limits.}
The declared operational limits used for normalization are
specified as follows. Let \(\tau_s^{\max}\) denote the
maximum tolerable C2 status-update staleness, \(d_{\min}\) the
minimum declared separation between reported UAV positions,
\(N_{\mathrm{op}}\) the maximum number of UAVs simultaneously
operating in the corridor, \(\rho_{\mathrm{max}}\) the maximum
population density within the declared operating area, and
\(\tau_v^{\max}\) the maximum tolerable coherence-window
staleness. Within these limits, we assume
\(h_{u,s}(t)\leq\tau_s^{\max}\),
\(N\leq N_{\mathrm{op}}\),
\(d_{uv}(t)\geq d_{\min}\),
\(\rho(\bm{x},t)\leq\rho_{\mathrm{max}}\), and
\(z_u(t)\leq\tau_v^{\max}\).

\noindent\textbf{Navigation Risk Normalization.}
At the maximum tolerable C2 status-update staleness, the
position-uncertainty radius in
\eqref{eq:uncertainty-radius} becomes \(r_{\tau}
=
r_{\min}
+
v_{\max}\tau_s^{\max}.\)
Under the declared operational limits, the unnormalized
navigation-risk component satisfies
\begin{align}
&\sum_{\substack{v=1\\v\neq u}}^{N}
\left[
r_u(t)+r_v(t)-d_{uv}(t)
\right]^+
\notag\\
&\quad\leq
(N_{\mathrm{op}}-1)
\left[
2r_{\tau}-d_{\min}
\right]^+.
\label{eq:navigation-reference-bound}
\end{align}
Accordingly, we define the reference navigation risk as
\begin{equation}
\bar R^{\mathrm{Nav}}
=
(N_{\mathrm{op}}-1)
\left[
2r_{\tau}-d_{\min}
\right]^+.
\label{eq:reference-navigation-risk}
\end{equation}
This is a conservative reference scale corresponding to the
simultaneous application of the declared fleet-size, C2
status-update staleness, and separation limits. We assume
\(d_{\min}<2r_{\tau}\), so that
\(\bar R^{\mathrm{Nav}}>0\).

\noindent\textbf{Ground Risk Normalization.}
At the maximum tolerable C2 status-update staleness, the
position-uncertainty region has area \(A_{\tau}
=
\pi r_{\tau}^{2}.\)
Since \(\rho(\bm{x},t)\leq\rho_{\mathrm{max}}\) and
\(r_u(t)\leq r_{\tau}\) under the declared operational
limits, the exposed population satisfies 
\begin{equation}
\int_{\mathcal B_u(t)}
\rho(\bm{x},t)\,\mathrm d\bm{x}
\leq
\rho_{\mathrm{max}}\pi r_{\tau}^{2}.
\label{eq:ground-reference-bound}
\end{equation}
We therefore define the reference ground risk as
\begin{equation}
\bar R^{\mathrm{Ground}}
=
\rho_{\mathrm{max}}A_{\tau}
=
\rho_{\mathrm{max}}\pi r_{\tau}^{2}.
\label{eq:reference-ground-risk}
\end{equation}
It represents a conservative exposed-population scale at the
declared C2 status-update staleness and population-density
limits.

\noindent\textbf{Video-Coherence Risk Normalization.}
Let \(H_{\mathrm{nom}}\) and \(\tau_v^{\max}\) denote the
nominal and maximum tolerable coherence-window staleness,
respectively, where
\(\tau_v^{\max}>H_{\mathrm{nom}}\). We define the reference
video-coherence risk as
\begin{equation}
\bar R^{\mathrm{Vid}}
=
\tau_v^{\max}
-
H_{\mathrm{nom}}.
\label{eq:reference-video-risk}
\end{equation}
Accordingly, the normalized video-coherence risk is zero for
\(z_u(t)\leq H_{\mathrm{nom}}\) and equals one when
\(z_u(t)=\tau_v^{\max}\).

Under the declared operational limits, each normalized risk
component is upper bounded by one. A value of one corresponds
to its conservative reference scale. A value greater than one
indicates that at least one of the limits used to construct
that reference scale has been exceeded.

\section{Derivation of the Composite Operational Risk Upper Bound}
\label{app:exposure-upper-bound}

In this appendix, we prove
Proposition~\ref{prop:risk-upper-bound}
and derive the coefficients
\(\alpha_{u,s}\),
\(\alpha_{u,v}\), and
\(C_u\).
We separately upper-bound the navigation, ground, and
video-coherence risk components. The local population-density
bound \(\rho_{u,\mathrm{max}}\) and pairwise separations
\(\{d_{uv}\}_{v\neq u}\) are determined by the operational
context considered in the analysis.

\noindent\textbf{Navigation Risk Upper Bound.}
Using the pairwise separation \(d_{uv}\) under the considered
operational context, together with \eqref{eq:navigation-risk}
and \eqref{eq:uncertainty-radius}, we have
\begin{align}
R_u^{\mathrm{Nav}}(t)
&=
\frac{1}{\bar R^{\mathrm{Nav}}}
\sum_{v\neq u}
\left[
r_u(t)+r_v(t)-d_{uv}
\right]^+
\notag\\
&=
\frac{1}{\bar R^{\mathrm{Nav}}}
\sum_{v\neq u}
\left[
2r_{\min}
+
v_{\max}h_{u,s}(t)
\right.
\notag\\
&\hspace{31mm}\left.
+
v_{\max}h_{v,s}(t)
-
d_{uv}
\right]^+.
\label{eq:app-navigation-distance-bound}
\end{align}
Under \eqref{eq:status-age-operating-region},
\(h_{v,s}(t)\leq\tau_{v,s}^{\max}\). Therefore,
\begin{align}
R_u^{\mathrm{Nav}}(t)
&\leq
\frac{1}{\bar R^{\mathrm{Nav}}}
\sum_{v\neq u}
\left[
v_{\max}h_{u,s}(t)+b_{uv}
\right]^+,
\label{eq:app-navigation-age-bound}
\end{align}
where
\begin{equation}
b_{uv}
=
2r_{\min}
+
v_{\max}\tau_{v,s}^{\max}
-
d_{uv}.
\label{eq:app-navigation-b-definition}
\end{equation}
Since \([v_{\max}h+b_{uv}]^+\) is convex in \(h\), its
secant line over
\(0\leq h\leq\tau_{u,s}^{\max}\) gives
\begin{equation}
\left[
v_{\max}h_{u,s}(t)+b_{uv}
\right]^+
\leq
\beta_{uv}h_{u,s}(t)
+
\gamma_{uv},
\label{eq:app-navigation-secant-bound}
\end{equation}
where
\begin{align}
\beta_{uv}
&=
\frac{
\left[
b_{uv}
+
v_{\max}\tau_{u,s}^{\max}
\right]^+
-
\left[
b_{uv}
\right]^+
}{
\tau_{u,s}^{\max}
},
\label{eq:app-navigation-beta-definition}
\\
\gamma_{uv}
&=
\left[
b_{uv}
\right]^+.
\label{eq:app-navigation-gamma-definition}
\end{align}
Substituting \eqref{eq:app-navigation-secant-bound} into
\eqref{eq:app-navigation-age-bound} yields
\begin{align}
R_u^{\mathrm{Nav}}(t)
&\leq
\frac{
\sum_{v\neq u}\beta_{uv}
}{
\bar R^{\mathrm{Nav}}
}
h_{u,s}(t)
+
\frac{
\sum_{v\neq u}\gamma_{uv}
}{
\bar R^{\mathrm{Nav}}
}.
\label{eq:app-navigation-bound}
\end{align}
After applying the navigation-risk weight \(w_N\), we obtain
\begin{align}
w_N R_u^{\mathrm{Nav}}(t)
&\leq
\frac{
w_N\sum_{v\neq u}\beta_{uv}
}{
\bar R^{\mathrm{Nav}}
}
h_{u,s}(t)
+
\frac{
w_N\sum_{v\neq u}\gamma_{uv}
}{
\bar R^{\mathrm{Nav}}
}.
\label{eq:app-weighted-navigation-bound}
\end{align}
Substituting \eqref{eq:app-navigation-b-definition} into
\eqref{eq:app-navigation-beta-definition} and
\eqref{eq:app-navigation-gamma-definition}, we obtain
\begin{align}
\beta_{uv}
&=
\frac{1}{\tau_{u,s}^{\max}}
\left(
\left[
2r_{\min}
+
v_{\max}
\left(
\tau_{u,s}^{\max}
+
\tau_{v,s}^{\max}
\right)
-
d_{uv}
\right]^+
\right.
\notag\\
&\hspace{21mm}\left.
-
\left[
2r_{\min}
+
v_{\max}\tau_{v,s}^{\max}
-
d_{uv}
\right]^+
\right),
\label{eq:app-navigation-beta-expanded}
\\
\gamma_{uv}
&=
\left[
2r_{\min}
+
v_{\max}\tau_{v,s}^{\max}
-
d_{uv}
\right]^+.
\label{eq:app-navigation-gamma-expanded}
\end{align}
Therefore, \eqref{eq:app-weighted-navigation-bound} becomes
\begin{align}
w_N R_u^{\mathrm{Nav}}(t)
&\leq
\frac{w_N}
{\bar R^{\mathrm{Nav}}\tau_{u,s}^{\max}}
\sum_{v\neq u}
\left(
\left[
2r_{\min}
+
v_{\max}
\left(
\tau_{u,s}^{\max}
+
\tau_{v,s}^{\max}
\right)
-
d_{uv}
\right]^+
\right.
\notag\\
&\hspace{22mm}\left.
-
\left[
2r_{\min}
+
v_{\max}\tau_{v,s}^{\max}
-
d_{uv}
\right]^+
\right)
h_{u,s}(t)
\notag\\
&\quad+
\frac{w_N}
{\bar R^{\mathrm{Nav}}}
\sum_{v\neq u}
\left[
2r_{\min}
+
v_{\max}\tau_{v,s}^{\max}
-
d_{uv}
\right]^+.
\label{eq:app-weighted-navigation-expanded}
\end{align}

\noindent\textbf{Ground Risk Upper Bound.}
From \eqref{eq:ground-risk} and the local population-density
bound in \eqref{eq:population-upper-bound}, we have
\begin{align}
R_u^{\mathrm{Ground}}(t)
&=
\frac{1}{\bar R^{\mathrm{Ground}}}
\int_{\mathcal B_u(t)}
\rho(\bm{x},t)\,\mathrm d\bm{x}
\notag\\
&\leq
\frac{
\rho_{u,\mathrm{max}}
}{
\bar R^{\mathrm{Ground}}
}
\left|
\mathcal B_u(t)
\right|.
\label{eq:app-ground-step1}
\end{align}
Using the area of the circular uncertainty region in
\eqref{eq:uncertainty-region-area}, we obtain
\begin{align}
R_u^{\mathrm{Ground}}(t)
&\leq
\frac{
\rho_{u,\mathrm{max}}\pi
}{
\bar R^{\mathrm{Ground}}
}
\left[
r_{\min}
+
v_{\max}h_{u,s}(t)
\right]^2.
\label{eq:app-ground-quadratic-bound}
\end{align}
Under \eqref{eq:status-age-operating-region},
\begin{equation}
h_{u,s}^{2}(t)
\leq
\tau_{u,s}^{\max}h_{u,s}(t).
\label{eq:app-status-age-quadratic-bound}
\end{equation}
Substituting \eqref{eq:app-status-age-quadratic-bound} into
\eqref{eq:app-ground-quadratic-bound} gives
\begin{align}
R_u^{\mathrm{Ground}}(t)
&\leq
\frac{
\rho_{u,\mathrm{max}}\pi
\left(
2r_{\min}v_{\max}
+
v_{\max}^{2}\tau_{u,s}^{\max}
\right)
}{
\bar R^{\mathrm{Ground}}
}
h_{u,s}(t)
\notag\\
&\quad+
\frac{
\rho_{u,\mathrm{max}}\pi r_{\min}^{2}
}{
\bar R^{\mathrm{Ground}}
}.
\label{eq:app-ground-linear-bound}
\end{align}
After applying the ground-risk weight \(w_G\), we obtain
\begin{align}
w_G R_u^{\mathrm{Ground}}(t)
&\leq
\frac{
w_G\rho_{u,\mathrm{max}}\pi
\left(
2r_{\min}v_{\max}
+
v_{\max}^{2}\tau_{u,s}^{\max}
\right)
}{
\bar R^{\mathrm{Ground}}
}
h_{u,s}(t)
\notag\\
&\quad+
\frac{
w_G\rho_{u,\mathrm{max}}\pi r_{\min}^{2}
}{
\bar R^{\mathrm{Ground}}
}.
\label{eq:app-weighted-ground-bound}
\end{align}

\noindent\textbf{Video-Coherence Risk Upper Bound.}
From \eqref{eq:video-risk}, since \(z_u(t)\geq0\),
\(H_{\mathrm{nom}}\geq0\), and
\(\bar R^{\mathrm{Vid}}>0\), we have
\begin{align}
R_u^{\mathrm{Vid}}(t)
&=
\left[
\frac{
z_u(t)-H_{\mathrm{nom}}
}{
\bar R^{\mathrm{Vid}}
}
\right]^+
\leq
\frac{
z_u(t)
}{
\bar R^{\mathrm{Vid}}
}.
\label{eq:app-video-bound}
\end{align}
Consequently,
\begin{equation}
w_V R_u^{\mathrm{Vid}}(t)
\leq
\frac{
w_V
}{
\bar R^{\mathrm{Vid}}
}
z_u(t).
\label{eq:app-weighted-video-bound}
\end{equation}
Since the completed-frame AoIs are ordered as \(h_{u,v}^{(1)}(t)
\leq
h_{u,v}^{(2)}(t)
\leq
\cdots
\leq
h_{u,v}^{(w)}(t),\)
and the coefficients \(\{\alpha_i\}_{i=1}^{w}\) are
nonnegative and satisfy
\(\sum_{i=1}^{w}\alpha_i=1\), we have
\begin{align}
z_u(t)
&=
\sum_{i=1}^{w}
\alpha_i h_{u,v}^{(i)}(t)
\leq
\sum_{i=1}^{w}
\alpha_i h_{u,v}^{(w)}(t)
=
h_{u,v}^{(w)}(t).
\label{eq:app-coherence-window-bound}
\end{align}
Combining
\eqref{eq:app-weighted-navigation-expanded},
\eqref{eq:app-weighted-ground-bound}, and
\eqref{eq:app-weighted-video-bound}
with the composite operational risk in
\eqref{eq:composite-risk}, and using the coefficient definitions
in \eqref{eq:alpha-status-definition}--%
\eqref{eq:risk-constant-definition}, we obtain
\begin{equation}
R_u(t)
\leq
\alpha_{u,s}h_{u,s}(t)
+
\alpha_{u,v}z_u(t)
+
C_u.
\label{eq:app-risk-bound-with-window}
\end{equation}
Finally, applying
\eqref{eq:app-coherence-window-bound} yields
\eqref{eq:risk-upper-bound}. This completes the proof of
Proposition~\ref{prop:risk-upper-bound}.

\section{PROOF OF THEOREM~\ref{thm:lower_bound}}
\label{proof_lowerbound}
Consider any admissible policy \(\pi\in\Pi\) operating over a
finite horizon of \(T\) slots.
For the status update flow of UAV \(u\), let \(D_{u,s}(T)\) denote
the number of successfully delivered status updates by the end
of slot \(T\). Let \(I_{u,s}^{(1)}\) denote the time from the beginning of
the horizon to the first successful status delivery, and, for
\(m\geq2\), let \(I_{u,s}^{(m)}\) denote the inter-delivery
time between the \((m-1)\)-th and \(m\)-th successful status
deliveries. Let \(R_{u,s}(T)\) denote the residual time after
the last successful status delivery. Then
\begin{equation}
T
=
\sum_{m=1}^{D_{u,s}(T)}
I_{u,s}^{(m)}
+
R_{u,s}(T).
\label{eq:status_partition_new}
\end{equation}
After a successful status delivery, the status AoI is reset to
\(1\). Hence, over an inter-delivery interval of length
\(I_{u,s}^{(m)}\), the AoI values are \(1,2,\ldots,I_{u,s}^{(m)}.\)
Therefore, the cumulative status AoI over this interval is
\begin{equation}
\sum_{\tau=1}^{I_{u,s}^{(m)}} \tau
=
\frac{I_{u,s}^{(m)}(I_{u,s}^{(m)}+1)}{2}
=
\frac{(I_{u,s}^{(m)})^2}{2}
+
\frac{I_{u,s}^{(m)}}{2}.    
\end{equation}
Applying this to all complete inter-delivery intervals and to
the residual interval, we obtain
\begin{align}
\sum_{t=1}^{T} h_{u,s}(t)
&=
\frac{1}{2}
\left[
\sum_{m=1}^{D_{u,s}(T)}
(I_{u,s}^{(m)})^2
+
R_{u,s}^2(T)
\right]
\notag\\
&\quad+
\frac{1}{2}
\left[
\sum_{m=1}^{D_{u,s}(T)}
I_{u,s}^{(m)}
+
R_{u,s}(T)
\right]
\notag\\
&=
\frac{1}{2}
\left[
\sum_{m=1}^{D_{u,s}(T)}
(I_{u,s}^{(m)})^2
+
R_{u,s}^2(T)
\right]
+
\frac{T}{2},
\label{eq:status_area_new}
\end{align}
where the last equality follows from
\eqref{eq:status_partition_new}.
If \(D_{u,s}(T)=0\), then no status update is delivered over
the horizon. Since \(h_{u,s}(1)=1\) and the status AoI
increases by one in every slot without a successful delivery,
\begin{equation}
\frac1T\sum_{t=1}^{T}h_{u,s}(t)
=
\frac{T+1}{2}
=
\frac12\frac{T}{D_{u,s}(T)+1}
+
\frac12.
\end{equation}
Thus, \eqref{eq:status_sample_path_new} holds directly when
\(D_{u,s}(T)=0\). It therefore suffices to consider
\(D_{u,s}(T)\geq1\) in the following derivation.

By Jensen's inequality,
\begin{equation}
\sum_{m=1}^{D_{u,s}(T)}
(I_{u,s}^{(m)})^2
\ge
D_{u,s}(T)
\left(
\frac{
\sum_{m=1}^{D_{u,s}(T)}
I_{u,s}^{(m)}
}{
D_{u,s}(T)
}
\right)^2 .
\label{eq:status_jensen_new}
\end{equation}
Using~\eqref{eq:status_partition_new}
we get
\begin{equation}
\sum_{m=1}^{D_{u,s}(T)}
(I_{u,s}^{(m)})^2
\ge
\frac{(T-R_{u,s}(T))^2}{D_{u,s}(T)}.
\label{eq:status_square_bound_new}
\end{equation}
Substituting \eqref{eq:status_square_bound_new} into
\eqref{eq:status_area_new} gives
\begin{equation}
\frac1T
\sum_{t=1}^{T} h_{u,s}(t)
\ge
\frac12
\left[
\frac{(T-R_{u,s}(T))^2}{TD_{u,s}(T)}
+
\frac{R_{u,s}^2(T)}{T}
\right]
+
\frac12.
\label{eq:status_before_residual_new}
\end{equation}
Minimizing the RHS of~\eqref{eq:status_before_residual_new} with respect to \(R_{u,s}(T)\) yields
\begin{equation}
\frac1T
\sum_{t=1}^{T} h_{u,s}(t)
\ge
\frac12
\frac{T}{D_{u,s}(T)+1}
+
\frac12.
\label{eq:status_sample_path_new}
\end{equation}
Taking expectations and using Jensen's inequality, we get
\begin{align}
\frac1T
\sum_{t=1}^{T}
\mathbb E^\pi[h_{u,s}(t)]
&\ge
\frac12
\mathbb E^\pi
\left[
\frac{T}{D_{u,s}(T)+1}
\right]
+
\frac12
\notag\\
&=
\frac12
\mathbb E^\pi
\left[
\frac{1}{D_{u,s}(T)/T+1/T}
\right]
+
\frac12
\notag\\
&\ge
\frac12
\frac{1}{
\mathbb E^\pi[D_{u,s}(T)]/T+1/T
}
+
\frac12.
\label{eq:status_expectation_new}
\end{align}
Define the finite-horizon expected status packet delivery rate as
\begin{equation}
q_{u,s}^{\pi}(T)
=
\frac{\mathbb E^\pi[D_{u,s}(T)]}{T}.
\label{eq:finite_status_rate}
\end{equation}
Then \eqref{eq:status_expectation_new} can be written as
\begin{equation}
\frac1T
\sum_{t=1}^{T}
\mathbb E^\pi[h_{u,s}(t)]
\ge
\frac{1}{
2\left(q_{u,s}^{\pi}(T)+1/T\right)
}
+
\frac12.
\label{eq:finite_status_age_bound}
\end{equation}

We now consider the video frame flow. Let \(D_{u,v}(T)\) denote
the number of completed video frames by the end of slot \(T\).
If \(D_{u,v}(T)=0\), then no video frame is completed over the
horizon. Since \(h_{u,v}(1)=1\) and the video AoI increases by one in
every slot without a frame completion,
\((1/T)\sum_{t=1}^{T}h_{u,v}(t)=(T+1)/2\).
Moreover, the RHS of
\eqref{eq:video_sample_path_new} evaluated at
\(D_{u,v}(T)=0\) is
\begin{equation}
\frac{(T-L+1)^2}{2T}
-
\frac{(L-1)^2}{2T}
+
L-\frac12
=
\frac{T+1}{2}.
\end{equation}
Thus, \eqref{eq:video_sample_path_new} holds directly when
\(D_{u,v}(T)=0\). It therefore suffices to consider
\(D_{u,v}(T)\geq1\) in the following derivation.

Let \(I_{u,v}^{(1)}\) denote the time from the beginning of
the horizon to the first video-frame completion, and, for
\(m\geq2\), let \(I_{u,v}^{(m)}\) denote the inter-completion
time between the \((m-1)\)-th and \(m\)-th completed video
frames. Let \(R_{u,v}(T)\) denote the residual time after the
last completed video frame. Then
\begin{equation}
T
=
\sum_{m=1}^{D_{u,v}(T)}
I_{u,v}^{(m)}
+
R_{u,v}(T).
\label{eq:video_partition_new}
\end{equation}

Before the first video-frame completion, the video AoI starts
at \(1\) and increases by one in every slot. Hence, over the
first interval of length \(I_{u,v}^{(1)}\), the video AoI
values are \(1,2,\ldots,I_{u,v}^{(1)}\).
After a video frame is completed, the video AoI is reset to
at least \(L\), because the frame contains \(L\) packets and
at most one packet can be successfully delivered in each
slot. Hence, over every subsequent inter-completion interval
of length \(I_{u,v}^{(m)}\), \(m\geq2\), the video AoI values
are lower bounded by
\(L,L+1,\ldots,L+I_{u,v}^{(m)}-1\).
Therefore, the cumulative video AoI over this interval is
lower bounded by
\begin{align}
\sum_{\tau=0}^{I_{u,v}^{(m)}-1}(L+\tau)
&=
\frac{(I_{u,v}^{(m)})^2}{2}
+
\left(L-\frac12\right)I_{u,v}^{(m)}.
\label{eq:video_interval_area_new}
\end{align}
The same lower bound applies to the residual interval after
the last completion. Consequently,
\begin{align}
\sum_{t=1}^{T}h_{u,v}(t)
&\ge
\frac12
\left[
\sum_{m=1}^{D_{u,v}(T)}
\left(I_{u,v}^{(m)}\right)^2
+
R_{u,v}^2(T)
\right]
\notag\\
&\quad+
\left(L-\frac12\right)T
-
(L-1)I_{u,v}^{(1)}.
\label{eq:video_area_corrected}
\end{align}
Completing the square in
\eqref{eq:video_area_corrected}, we obtain
\begin{align}
\sum_{t=1}^{T}h_{u,v}(t)
&\ge
\frac12
\left[
\left(I_{u,v}^{(1)}-(L-1)\right)^2
+
\sum_{m=2}^{D_{u,v}(T)}
\left(I_{u,v}^{(m)}\right)^2
\right]
\notag\\
&\quad+
\frac12R_{u,v}^2(T)
-
\frac{(L-1)^2}{2}
+
\left(L-\frac12\right)T.
\label{eq:video_area_square}
\end{align}
By Jensen's inequality,
\begin{align}
&\left(I_{u,v}^{(1)}-(L-1)\right)^2
+
\sum_{m=2}^{D_{u,v}(T)}
\left(I_{u,v}^{(m)}\right)^2
+
R_{u,v}^2(T)
\notag\\
&\ge
\frac{1}{D_{u,v}(T)+1}
\left[
I_{u,v}^{(1)}-(L-1)
+
\sum_{m=2}^{D_{u,v}(T)}
I_{u,v}^{(m)}
\right.
\notag\\
&\hspace{42mm}\left.
+
R_{u,v}(T)
\right]^2.
\label{eq:video_jensen_new}
\end{align}
Using~\eqref{eq:video_partition_new}, we get
\begin{align}
&\left(I_{u,v}^{(1)}-(L-1)\right)^2
+
\sum_{m=2}^{D_{u,v}(T)}
\left(I_{u,v}^{(m)}\right)^2
+
R_{u,v}^2(T)
\notag\\
&\quad\ge
\frac{(T-L+1)^2}{D_{u,v}(T)+1}.
\label{eq:video_square_bound_new}
\end{align}
Substituting \eqref{eq:video_square_bound_new} into
\eqref{eq:video_area_square} gives
\begin{align}
\frac1T\sum_{t=1}^{T}h_{u,v}(t)
&\ge
\frac{(T-L+1)^2}
{2T\left(D_{u,v}(T)+1\right)}
-
\frac{(L-1)^2}{2T}
+
L-\frac12.
\label{eq:video_sample_path_new}
\end{align}
Taking expectations and applying Jensen's inequality gives
\begin{align}
\frac1T
\sum_{t=1}^{T}
\mathbb E^\pi[h_{u,v}(t)]
&\ge
\frac{
\left(1-(L-1)/T\right)^2
}{
2\left(
\mathbb E^\pi[D_{u,v}(T)]/T+1/T
\right)
}
\notag\\
&\quad-
\frac{(L-1)^2}{2T}
+
L-\frac12.
\label{eq:video_expectation_new}
\end{align}
Define the finite-horizon expected video-frame completion
rate as
\begin{equation}
q_{u,v}^{\pi}(T)
=
\frac{\mathbb E^\pi[D_{u,v}(T)]}{T}.
\label{eq:finite_video_rate}
\end{equation}
Then \eqref{eq:video_expectation_new} becomes
\begin{align}
\frac1T
\sum_{t=1}^{T}
\mathbb E^\pi[h_{u,v}(t)]
&\ge
\frac{
\left(1-(L-1)/T\right)^2
}{
2\left(q_{u,v}^{\pi}(T)+1/T\right)
}
\notag\\
&\quad-
\frac{(L-1)^2}{2T}
+
L-\frac12.
\label{eq:finite_video_age_bound}
\end{align}

It remains to characterize the finite-horizon delivery-rate
vectors. For the status update flow,
\begin{align}
\mathbb E^\pi[D_{u,s}(T)]
&=
\sum_{t=1}^{T}
\mathbb E^\pi[d_{u,s}(t)]
=
p_{u,s}
\sum_{t=1}^{T}
\mathbb E^\pi[a_{u,s}(t)].
\label{eq:finite_status_delivery}
\end{align}
Therefore,
\begin{equation}
\frac1T
\sum_{t=1}^{T}
\mathbb E^\pi[a_{u,s}(t)]= \frac{q_{u,s}^{\pi}(T)}{p_{u,s}}.
\label{eq:finite_status_scheduling}
\end{equation}
For the video frame flow,
\begin{equation}
p_{u,v}
\sum_{t=1}^{T}
\mathbb E^\pi[a_{u,v}(t)]
=
\sum_{t=1}^{T}
\mathbb E^\pi[\mu_{u,v}(t)].
\label{eq:finite_video_successes}
\end{equation}
Every completed frame contains exactly \(L\) successfully
delivered packets. Hence,
\begin{equation}
\sum_{t=1}^{T}\mu_{u,v}(t)
=
L D_{u,v}(T)+O_{u,v}(T),
\label{eq:video_packet_partition}
\end{equation}
where \(0\leq O_{u,v}(T)\leq L-1\) is the number of
successfully delivered packets belonging to the unfinished
frame at the end of slot \(T\). It follows that
\begin{equation}
\frac1T
\sum_{t=1}^{T}
\mathbb E^\pi[a_{u,v}(t)]
\ge
\frac{Lq_{u,v}^{\pi}(T)}{p_{u,v}}.
\label{eq:finite_video_scheduling}
\end{equation}

Taking expectations and time averages in
\eqref{eq:resource_constraint}, and using
\eqref{eq:finite_status_scheduling} and
\eqref{eq:finite_video_scheduling}, gives, for every \(T\),
\begin{equation}
\sum_{u=1}^{N}
\left(
\frac{q_{u,s}^{\pi}(T)}{p_{u,s}}
+
\frac{Lq_{u,v}^{\pi}(T)}{p_{u,v}}
\right)
\leq K.
\label{eq:finite_rate_constraint}
\end{equation}
Let \(\bm q^\pi(T)
=
\left(
\{q_{u,s}^{\pi}(T)\}_{u=1}^{N},
\{q_{u,v}^{\pi}(T)\}_{u=1}^{N}
\right)\).
Since \(0\leq q_{u,s}^{\pi}(T)\leq p_{u,s}\) and
\(0\leq q_{u,v}^{\pi}(T)\leq p_{u,v}/L\), the sequence
\(\{\bm q^\pi(T)\}_{T\geq1}\) is bounded. Therefore, there
exists a convergent subsequence indexed by a sequence of
horizon lengths \(\{T_n\}\), with \(T_n\to\infty\), and a
vector
\(\bm q^\pi
=
\left(
\{q_{u,s}^{\pi}\}_{u=1}^{N},
\{q_{u,v}^{\pi}\}_{u=1}^{N}
\right)\)
such that \(\bm q^\pi(T_n)\to\bm q^\pi\). Taking
\(n\to\infty\) in \eqref{eq:finite_rate_constraint} shows
that this limit point satisfies
\begin{equation}
\sum_{u=1}^{N}
\left(
\frac{q_{u,s}^{\pi}}{p_{u,s}}
+
\frac{Lq_{u,v}^{\pi}}{p_{u,v}}
\right)
\leq K.
\label{eq:rate_resource_constraint_new}
\end{equation}
Combining \eqref{eq:finite_status_age_bound} and
\eqref{eq:finite_video_age_bound}, and taking the limit along
the subsequence \(\{T_n\}\), gives
\begin{align}
\widehat J^\pi
&=
\limsup_{T\to\infty}
\frac{1}{T}
\sum_{t=1}^{T}
\sum_{u=1}^{N}
\mathbb E^\pi
\left[
\alpha_{u,s}h_{u,s}(t)
+
\alpha_{u,v}h_{u,v}(t)
\right]
\notag\\
&\ge
\liminf_{n\to\infty}
\frac{1}{T_n}
\sum_{t=1}^{T_n}
\sum_{u=1}^{N}
\mathbb E^\pi
\left[
\alpha_{u,s}h_{u,s}(t)
+
\alpha_{u,v}h_{u,v}(t)
\right]
\notag\\
&\ge
\sum_{u=1}^{N}
\alpha_{u,s}
\left(
\frac{1}{2q_{u,s}^{\pi}}
+
\frac{1}{2}
\right)
+
\sum_{u=1}^{N}
\alpha_{u,v}
\left(
\frac{1}{2q_{u,v}^{\pi}}
+
L-\frac{1}{2}
\right).
\label{eq:risk_rate_bound_new}
\end{align}
Interpreting \(1/0\) as \(+\infty\), if any component of
\(\bm q^\pi\) is zero, \eqref{eq:risk_rate_bound_new} implies
\(\widehat J^\pi=+\infty\), and hence
\(\widehat J^\pi\geq L_B^\star\) holds directly.
Otherwise, all components of \(\bm q^\pi\) are strictly
positive. Together with
\eqref{eq:rate_resource_constraint_new} and \(q_{u,s}^{\pi}\leq p_{u,s},
q_{u,v}^{\pi}\leq {p_{u,v}}/{L},\)
this shows that \(\bm q^\pi\) is feasible for the minimization
in \eqref{eq:LB}. Therefore, the RHS of
\eqref{eq:risk_rate_bound_new} is no smaller than
\(L_B^\star\), and hence
\(\widehat J^\pi\geq L_B^\star\).
Since this inequality holds for every admissible policy
\(\pi\in\Pi\), we conclude that
\(\widehat J^\star\geq L_B^\star\).

\section{Upper Bound on the One-Slot Lyapunov Drift}
\label{app:lyapunov_drift_linear}

In this appendix, we derive the one-slot Lyapunov drift upper
bound, the status scheduling weight, and the video scheduling
weight given in \eqref{eq:MW_drift_bound_linear},
\eqref{eq:MW_status_weight_linear}, and
\eqref{eq:MW_video_weight_linear}, respectively.
Recall the system state \(S(t)\), the linear video potential
\(\Phi_{u,v}(t)\) defined in
\eqref{eq:video_linear_potential}, the Lyapunov function
\eqref{eq:MW_lyapunov_linear}, and the one-slot conditional
Lyapunov drift \eqref{eq:MW_drift_linear}.
Substituting \eqref{eq:MW_lyapunov_linear} into
\eqref{eq:MW_drift_linear}, we obtain
\begin{align}
\Delta(S(t))
&=
\sum_{u=1}^{N}
V_{u,s}
\mathbb E\!\left[
h_{u,s}(t+1)-h_{u,s}(t)
\mid S(t)
\right]
\notag\\
&\quad+
\sum_{u=1}^{N}
V_{u,v}
\mathbb E\!\left[
\Phi_{u,v}(t+1)-\Phi_{u,v}(t)
\mid S(t)
\right].
\label{eq:app_linear_drift_expand}
\end{align}

\noindent\textbf{Status Evolution.}
The status AoI evolves as
\begin{equation}
h_{u,s}(t+1)
=
h_{u,s}(t)+1-d_{u,s}(t)h_{u,s}(t).
\label{eq:app_linear_status_evolution}
\end{equation}
Taking the conditional expectation and using \(\mathbb E[d_{u,s}(t)\mid S(t)]
=
p_{u,s}
\mathbb E[a_{u,s}(t)\mid S(t)],\)
we obtain
\begin{align}
&\mathbb E\!\left[
h_{u,s}(t+1)-h_{u,s}(t)
\mid S(t)
\right]
\notag\\
&\quad=
1
-
p_{u,s}h_{u,s}(t)
\mathbb E[a_{u,s}(t)\mid S(t)].
\label{eq:app_linear_status_drift}
\end{align}

\noindent\textbf{Video Evolution.}
Recall that \(\Phi_{u,v}(t)
=
({\ell_u(t)}/{L})(h_{u,v}(t)-1)
+
g_u(t)
+
1\).
We consider two cases.

\noindent
\emph{Case 1: \(\ell_u(t)=L\).}
Under the generate-at-will model, \(g_u(t)=0\) whenever
\(\ell_u(t)=L\). If the video packet is not successfully
delivered, i.e., \(\mu_{u,v}(t)=0\), then
\[
\left\{
\begin{aligned}
h_{u,v}(t+1) &= h_{u,v}(t)+1,\\
g_u(t+1) &= 0,\\
\ell_u(t+1) &= L.
\end{aligned}
\right.
\]
Hence,
\begin{equation}
\Phi_{u,v}(t+1)-\Phi_{u,v}(t)
=
1.
\label{eq:app_linear_video_fail_initial}
\end{equation}
If the video packet is successfully delivered, i.e.,
\(\mu_{u,v}(t)=1\), and \(L>1\), then
\[
\left\{
\begin{aligned}
h_{u,v}(t+1) &= h_{u,v}(t)+1,\\
g_u(t+1) &= g_u(t)+1=1,\\
\ell_u(t+1) &= L-1.
\end{aligned}
\right.
\]
Hence,
\begin{equation}
\Phi_{u,v}(t+1)-\Phi_{u,v}(t)
=
2-\frac{h_{u,v}(t)}{L}.
\label{eq:app_linear_video_success_initial}
\end{equation}
For \(L=1\), the successful transmission completes the frame, and
\[
\left\{
\begin{aligned}
h_{u,v}(t+1) &= g_u(t)+1=1,\\
g_u(t+1) &= 0,\\
\ell_u(t+1) &= L.
\end{aligned}
\right.
\]
Hence, \(\Phi_{u,v}(t+1)-\Phi_{u,v}(t)
=
1-h_{u,v}(t)
\leq
2-h_{u,v}(t).\)
Therefore, when \(\ell_u(t)=L\),
\begin{equation}
\Phi_{u,v}(t+1)-\Phi_{u,v}(t)
\leq
1+\frac{\ell_u(t)}{L}
-
\frac{\mu_{u,v}(t)}{L}h_{u,v}(t).
\label{eq:app_linear_video_initial_bound}
\end{equation}

\noindent
\emph{Case 2: \(\ell_u(t)<L\).}
If the video packet is not successfully delivered, i.e.,
\(\mu_{u,v}(t)=0\), then
\[
\left\{
\begin{aligned}
h_{u,v}(t+1) &= h_{u,v}(t)+1,\\
g_u(t+1) &= g_u(t)+1,\\
\ell_u(t+1) &= \ell_u(t).
\end{aligned}
\right.
\]
Hence,
\begin{equation}
\Phi_{u,v}(t+1)-\Phi_{u,v}(t)
=
1+\frac{\ell_u(t)}{L}.
\label{eq:app_linear_video_fail_started}
\end{equation}
If \(\mu_{u,v}(t)=1\) and \(\ell_u(t)>1\), then
\[
\left\{
\begin{aligned}
h_{u,v}(t+1) &= h_{u,v}(t)+1,\\
g_u(t+1) &= g_u(t)+1,\\
\ell_u(t+1) &= \ell_u(t)-1.
\end{aligned}
\right.
\]
Hence,
\begin{equation}
\Phi_{u,v}(t+1)-\Phi_{u,v}(t)
=
1+\frac{\ell_u(t)}{L}
-
\frac{h_{u,v}(t)}{L}.
\label{eq:app_linear_video_success_started}
\end{equation}
If \(\mu_{u,v}(t)=1\) and \(\ell_u(t)=1\), then the frame is
completed and
\[
\left\{
\begin{aligned}
h_{u,v}(t+1) &= g_u(t)+1,\\
g_u(t+1) &= 0,\\
\ell_u(t+1) &= L.
\end{aligned}
\right.
\]
Hence,
\begin{equation}
\Phi_{u,v}(t+1)-\Phi_{u,v}(t)
=
\frac{1}{L}
-
\frac{h_{u,v}(t)}{L}.
\label{eq:app_linear_video_success_completion}
\end{equation}
Therefore, when \(\ell_u(t)<L\),
\begin{equation}
\Phi_{u,v}(t+1)-\Phi_{u,v}(t)
\leq
1+\frac{\ell_u(t)}{L}
-
\frac{\mu_{u,v}(t)}{L}h_{u,v}(t).
\label{eq:app_linear_video_started_increment}
\end{equation}

Combining the two cases, for every
\(\ell_u(t)\in\{1,\ldots,L\}\),
\begin{equation}
\Phi_{u,v}(t+1)-\Phi_{u,v}(t)
\leq
1+\frac{\ell_u(t)}{L}
-
\frac{\mu_{u,v}(t)}{L}h_{u,v}(t).
\label{eq:app_linear_video_increment}
\end{equation}
Taking the conditional expectation and using \(\mathbb E[\mu_{u,v}(t)\mid S(t)]
=
p_{u,v}
\mathbb E[a_{u,v}(t)\mid S(t)],\)
we obtain
\begin{align}
&\mathbb E\!\left[
\Phi_{u,v}(t+1)-\Phi_{u,v}(t)
\mid S(t)
\right]
\notag\\
&\quad\leq
1+\frac{\ell_u(t)}{L}
-
\frac{p_{u,v}}{L}h_{u,v}(t)
\mathbb E[a_{u,v}(t)\mid S(t)].
\label{eq:app_linear_video_drift}
\end{align}
Substituting
\eqref{eq:app_linear_status_drift} and
\eqref{eq:app_linear_video_drift} into
\eqref{eq:app_linear_drift_expand}, we obtain
\begin{align}
\Delta(S(t))
&\leq
\sum_{u=1}^{N}V_{u,s}
+
\sum_{u=1}^{N}
V_{u,v}
\left(
1+\frac{\ell_u(t)}{L}
\right)
\notag\\
&\quad-
\sum_{u=1}^{N}
V_{u,s}p_{u,s}h_{u,s}(t)
\mathbb E\!\left[
a_{u,s}(t)\mid S(t)
\right]
\notag\\
&\quad-
\sum_{u=1}^{N}
\frac{V_{u,v}p_{u,v}}{L}
h_{u,v}(t)
\mathbb E\!\left[
a_{u,v}(t)\mid S(t)
\right].
\label{eq:app_linear_drift_expanded}
\end{align}
Rearranging the terms in
\eqref{eq:app_linear_drift_expanded} yields
\eqref{eq:MW_drift_bound_linear}--%
\eqref{eq:MW_video_weight_linear}.

\section{Stationary Randomized Policy}
\label{app:randomized-comparison-policy}

This appendix introduces the class of stationary randomized
policies used as comparison policies in the proof of
Theorem~\ref{thm:MW_linear_ratio}.

Define the set of feasible scheduling actions as
\begin{equation}
\mathcal A
\triangleq
\left\{
\bm a\in\{0,1\}^{2N}:
\sum_{u=1}^{N}
\left(
a_{u,s}+a_{u,v}
\right)
\leq K
\right\}.
\label{eq:randomized-action-set}
\end{equation}
A stationary randomized policy independently selects an action
\(\bm a(t)\in\mathcal A\) in every slot according to a fixed
probability distribution over \(\mathcal A\), independently of
the current system state and previous scheduling decisions.

Let \(\theta_{u,s}\) and \(\theta_{u,v}\) denote the marginal
probabilities of scheduling the status update and video frame flows of UAV
\(u\), respectively:
\begin{equation}
\theta_{u,s}
=
\mathbb E[a_{u,s}(t)],
\qquad
\theta_{u,v}
=
\mathbb E[a_{u,v}(t)].
\label{eq:randomized-marginals}
\end{equation}
These marginal probabilities satisfy
\begin{equation}
\sum_{u=1}^{N}
\left(
\theta_{u,s}
+
\theta_{u,v}
\right)
\leq K,
\qquad
0\leq\theta_{u,s},\theta_{u,v}\leq1,
\quad \forall u.
\label{eq:randomized-marginal-constraint}
\end{equation}

Conversely, every vector satisfying
\eqref{eq:randomized-marginal-constraint} belongs to the convex
hull of \(\mathcal A\). Therefore, it can be realized by a
probability distribution over feasible actions in \(\mathcal A\).
For the resulting stationary randomized policy,
\begin{equation}
\mathbb E[a_{u,s}(t)\mid S(t)]
=
\theta_{u,s},
\qquad
\mathbb E[a_{u,v}(t)\mid S(t)]
=
\theta_{u,v},
\label{eq:randomized-conditional-actions}
\end{equation}
for every state \(S(t)\).

Let
\(\{q_{u,s}^{\mathrm{LB}},q_{u,v}^{\mathrm{LB}}\}_{u=1}^{N}\)
be an optimal solution of the lower-bound problem
\eqref{eq:LB}. Define
\begin{equation}
\theta_{u,s}^{\mathrm{LB}}
=
\frac{q_{u,s}^{\mathrm{LB}}}{p_{u,s}},
\qquad
\theta_{u,v}^{\mathrm{LB}}
=
\frac{Lq_{u,v}^{\mathrm{LB}}}{p_{u,v}}.
\label{eq:comparison-policy-rates}
\end{equation}
The rate constraints in the lower-bound problem imply
\begin{equation}
0
<
\theta_{u,s}^{\mathrm{LB}}
\leq1,
\qquad
0
<
\theta_{u,v}^{\mathrm{LB}}
\leq1,
\quad \forall u.
\label{eq:comparison-policy-individual-feasibility}
\end{equation}
and
\begin{align}
\sum_{u=1}^{N}
\left(
\theta_{u,s}^{\mathrm{LB}}
+
\theta_{u,v}^{\mathrm{LB}}
\right)
&=
\sum_{u=1}^{N}
\left(
\frac{q_{u,s}^{\mathrm{LB}}}{p_{u,s}}
+
\frac{Lq_{u,v}^{\mathrm{LB}}}{p_{u,v}}
\right)
\leq K.
\label{eq:comparison-policy-feasibility}
\end{align}
Hence, the lower-bound solution induces a feasible stationary
randomized comparison policy with marginal scheduling
probabilities
\(\{\theta_{u,s}^{\mathrm{LB}},
\theta_{u,v}^{\mathrm{LB}}\}_{u=1}^{N}\).

\section{Proof of Theorem~\ref{thm:MW_linear_ratio}}
\label{app:MW_linear_ratio_proof}

Let
\(\{q_{u,s}^{\mathrm{LB}},q_{u,v}^{\mathrm{LB}}\}_{u=1}^{N}\)
be an optimal solution of the lower-bound problem
\eqref{eq:LB}. As shown in
Appendix~\ref{app:randomized-comparison-policy}, this solution
induces a feasible stationary randomized comparison policy with
marginal scheduling probabilities
\begin{equation}
\theta_{u,s}^{\mathrm{LB}}
=
\frac{q_{u,s}^{\mathrm{LB}}}{p_{u,s}},
\qquad
\theta_{u,v}^{\mathrm{LB}}
=
\frac{Lq_{u,v}^{\mathrm{LB}}}{p_{u,v}}.
\label{eq:comparison_policy_rates_ratio}
\end{equation}

From the one-slot drift upper bound
\eqref{eq:MW_drift_bound_linear}, under the Max-Weight policy,
\begin{align}
\Delta^{\mathrm{MW}}(S(t))
&\leq
B(t)
-
\sum_{u=1}^{N}
W_{u,s}(t)
\mathbb E^{\mathrm{MW}}\!\left[
a_{u,s}(t)\mid S(t)
\right]
\notag\\
&\quad-
\sum_{u=1}^{N}
W_{u,v}(t)
\mathbb E^{\mathrm{MW}}\!\left[
a_{u,v}(t)\mid S(t)
\right].
\label{eq:MW_ratio_drift_start}
\end{align}
Since the Max-Weight policy maximizes the total scheduling
weight in every slot, for every state \(S(t)\),
\begin{align}
&
\sum_{u=1}^{N}
W_{u,s}(t)
\mathbb E^{\mathrm{MW}}\!\left[
a_{u,s}(t)\mid S(t)
\right]
\notag\\
&\quad+
\sum_{u=1}^{N}
W_{u,v}(t)
\mathbb E^{\mathrm{MW}}\!\left[
a_{u,v}(t)\mid S(t)
\right]
\notag\\
&\geq
\sum_{u=1}^{N}
W_{u,s}(t)\theta_{u,s}^{\mathrm{LB}}
+
\sum_{u=1}^{N}
W_{u,v}(t)\theta_{u,v}^{\mathrm{LB}}.
\label{eq:MW_dominates_randomized}
\end{align}
Substituting \eqref{eq:MW_dominates_randomized} into
\eqref{eq:MW_ratio_drift_start} gives
\begin{align}
\Delta^{\mathrm{MW}}(S(t))
&\leq
B(t)
-
\sum_{u=1}^{N}
W_{u,s}(t)\theta_{u,s}^{\mathrm{LB}}
-
\sum_{u=1}^{N}
W_{u,v}(t)\theta_{u,v}^{\mathrm{LB}}.
\label{eq:MW_ratio_compare_drift}
\end{align}

Using the Max-Weight weights in
\eqref{eq:MW_status_weight_linear} and
\eqref{eq:MW_video_weight_linear}, together with
\eqref{eq:comparison_policy_rates_ratio}, we obtain
\begin{align}
W_{u,s}(t)\theta_{u,s}^{\mathrm{LB}}
&=
V_{u,s}p_{u,s}h_{u,s}(t)
\frac{q_{u,s}^{\mathrm{LB}}}{p_{u,s}}
=
V_{u,s}q_{u,s}^{\mathrm{LB}}h_{u,s}(t),
\label{eq:status_weight_compare}
\\
W_{u,v}(t)\theta_{u,v}^{\mathrm{LB}}
&=
\frac{V_{u,v}p_{u,v}}{L}h_{u,v}(t)
\frac{Lq_{u,v}^{\mathrm{LB}}}{p_{u,v}}
=
V_{u,v}q_{u,v}^{\mathrm{LB}}h_{u,v}(t).
\label{eq:video_weight_compare}
\end{align}
Therefore,
\begin{align}
\Delta^{\mathrm{MW}}(S(t))
&\leq
B(t)
-
\sum_{u=1}^{N}
V_{u,s}q_{u,s}^{\mathrm{LB}}h_{u,s}(t)
\notag\\
&\quad-
\sum_{u=1}^{N}
V_{u,v}q_{u,v}^{\mathrm{LB}}h_{u,v}(t).
\label{eq:MW_ratio_drift_rate}
\end{align}
Choose
\(V_{u,s}
=
\alpha_{u,s}/q_{u,s}^{\mathrm{LB}}\)
and
\(V_{u,v}
=
\alpha_{u,v}/q_{u,v}^{\mathrm{LB}}\).
Then
\(V_{u,s}q_{u,s}^{\mathrm{LB}}=\alpha_{u,s}\)
and
\(V_{u,v}q_{u,v}^{\mathrm{LB}}=\alpha_{u,v}\).
Hence,
\begin{equation}
\Delta^{\mathrm{MW}}(S(t))
\leq
B(t)
-
\sum_{u=1}^{N}
\left(
\alpha_{u,s}h_{u,s}(t)
+
\alpha_{u,v}h_{u,v}(t)
\right).
\label{eq:MW_ratio_drift_risk}
\end{equation}
From \eqref{eq:MW_B_term_linear} and
\(1\leq\ell_u(t)\leq L\),
\begin{align}
B(t)
&=
\sum_{u=1}^{N}V_{u,s}
+
\sum_{u=1}^{N}
V_{u,v}
\left(
1+\frac{\ell_u(t)}{L}
\right)
\notag\\
&\leq
\sum_{u=1}^{N}V_{u,s}
+
2\sum_{u=1}^{N}V_{u,v}
=
\sum_{u=1}^{N}
\frac{\alpha_{u,s}}{q_{u,s}^{\mathrm{LB}}}
+
2\sum_{u=1}^{N}
\frac{\alpha_{u,v}}{q_{u,v}^{\mathrm{LB}}}.
\label{eq:MW_B_constant}
\end{align}
Substituting \eqref{eq:MW_B_constant} into
\eqref{eq:MW_ratio_drift_risk} gives
\begin{align}
&\Delta^{\mathrm{MW}}(S(t))
+
\sum_{u=1}^{N}
\left(
\alpha_{u,s}h_{u,s}(t)
+
\alpha_{u,v}h_{u,v}(t)
\right)
\notag\\
&\leq
\sum_{u=1}^{N}
\frac{\alpha_{u,s}}{q_{u,s}^{\mathrm{LB}}}
+
2\sum_{u=1}^{N}
\frac{\alpha_{u,v}}{q_{u,v}^{\mathrm{LB}}}.
\label{eq:MW_drift_plus_risk}
\end{align}
Following \eqref{eq:surrogate_system_risk}, let
\(\widehat J^{\mathrm{MW}}\) denote the long-term expected
surrogate system risk achieved by the Max-Weight policy.
Summing \eqref{eq:MW_drift_plus_risk} over
\(t=1,\ldots,T\), taking expectations, and telescoping the
Lyapunov terms yields
\begin{align}
&
\sum_{t=1}^{T}
\sum_{u=1}^{N}
\mathbb E^{\mathrm{MW}}\!\left[
\alpha_{u,s}h_{u,s}(t)
+
\alpha_{u,v}h_{u,v}(t)
\right]
\notag\\
&\leq
T
\left(
\sum_{u=1}^{N}
\frac{\alpha_{u,s}}{q_{u,s}^{\mathrm{LB}}}
+
2\sum_{u=1}^{N}
\frac{\alpha_{u,v}}{q_{u,v}^{\mathrm{LB}}}
\right)
\notag\\
&\quad+
\mathbb E^{\mathrm{MW}}[L(S(1))]
-
\mathbb E^{\mathrm{MW}}[L(S(T+1))].
\label{eq:MW_telescoping}
\end{align}
Since \(L(S(T+1))\geq0\) and
\(\mathbb E^{\mathrm{MW}}[L(S(1))]<\infty\), dividing by
\(T\) and taking the limit superior as \(T\to\infty\) gives
\begin{equation}
\widehat J^{\mathrm{MW}}
\leq
\sum_{u=1}^{N}
\frac{\alpha_{u,s}}{q_{u,s}^{\mathrm{LB}}}
+
2\sum_{u=1}^{N}
\frac{\alpha_{u,v}}{q_{u,v}^{\mathrm{LB}}}.
\label{eq:MW_risk_upper_final}
\end{equation}
It remains to compare this upper bound with \(L_B^\star\).
From the lower-bound problem,
\begin{align}
L_B^\star
&=
\frac{1}{2}
\sum_{u=1}^{N}
\left(
\frac{\alpha_{u,s}}{q_{u,s}^{\mathrm{LB}}}
+
\frac{\alpha_{u,v}}{q_{u,v}^{\mathrm{LB}}}
\right)
\notag\\
&\quad+
\frac{1}{2}
\sum_{u=1}^{N}
\left(
\alpha_{u,s}
+
(2L-1)\alpha_{u,v}
\right).
\label{eq:LB_exact_in_MW_proof}
\end{align}
Therefore,
\begin{align}
&
\sum_{u=1}^{N}
\frac{\alpha_{u,s}}{q_{u,s}^{\mathrm{LB}}}
+
2\sum_{u=1}^{N}
\frac{\alpha_{u,v}}{q_{u,v}^{\mathrm{LB}}}
\notag\\
&=
L_B^\star
+
\frac{1}{2}
\sum_{u=1}^{N}
\frac{\alpha_{u,s}}{q_{u,s}^{\mathrm{LB}}}
+
\frac{3}{2}
\sum_{u=1}^{N}
\frac{\alpha_{u,v}}{q_{u,v}^{\mathrm{LB}}}
\notag\\
&\quad-
\frac{1}{2}
\sum_{u=1}^{N}
\left(
\alpha_{u,s}
+
(2L-1)\alpha_{u,v}
\right)
\notag\\
&\leq
L_B^\star
+
\frac{1}{2}
\sum_{u=1}^{N}
\frac{\alpha_{u,s}}{q_{u,s}^{\mathrm{LB}}}
+
\frac{3}{2}
\sum_{u=1}^{N}
\frac{\alpha_{u,v}}{q_{u,v}^{\mathrm{LB}}},
\label{eq:MW_upper_rewrite}
\end{align}
where the inequality follows from
\(\alpha_{u,s}\geq0\), \(\alpha_{u,v}\geq0\), and \(L\geq1\).
Combining \eqref{eq:MW_risk_upper_final} and
\eqref{eq:MW_upper_rewrite}, we obtain
\begin{equation}
\widehat J^{\mathrm{MW}}
\leq
L_B^\star
+
\frac{1}{2}
\sum_{u=1}^{N}
\frac{\alpha_{u,s}}{q_{u,s}^{\mathrm{LB}}}
+
\frac{3}{2}
\sum_{u=1}^{N}
\frac{\alpha_{u,v}}{q_{u,v}^{\mathrm{LB}}},
\label{eq:MW_linear_upper_bound_proof}
\end{equation}
which establishes \eqref{eq:MW_linear_upper_bound}.

Finally, by Theorem~\ref{thm:lower_bound},
\(L_B^\star\leq\widehat J^\star\). Therefore,
\begin{align}
\psi^{\mathrm{MW}}
=
\frac{\widehat J^{\mathrm{MW}}}{\widehat J^\star}
&\leq
\frac{
\displaystyle
L_B^\star
+
\frac{1}{2}
\sum_{u=1}^{N}
\frac{\alpha_{u,s}}{q_{u,s}^{\mathrm{LB}}}
+
\frac{3}{2}
\sum_{u=1}^{N}
\frac{\alpha_{u,v}}{q_{u,v}^{\mathrm{LB}}}
}{
\widehat J^\star
}
\notag\\
&\leq
\frac{
\displaystyle
L_B^\star
+
\frac{1}{2}
\sum_{u=1}^{N}
\frac{\alpha_{u,s}}{q_{u,s}^{\mathrm{LB}}}
+
\frac{3}{2}
\sum_{u=1}^{N}
\frac{\alpha_{u,v}}{q_{u,v}^{\mathrm{LB}}}
}{
L_B^\star
}
\notag\\
&=
1+
\frac{
\displaystyle
\frac{1}{2}
\sum_{u=1}^{N}
\frac{\alpha_{u,s}}{q_{u,s}^{\mathrm{LB}}}
+
\frac{3}{2}
\sum_{u=1}^{N}
\frac{\alpha_{u,v}}{q_{u,v}^{\mathrm{LB}}}
}{
L_B^\star
}.
\label{eq:MW_ratio_proof}
\end{align}
This establishes \eqref{eq:MW_linear_ratio_bound} and completes
the proof.

\section{Real-World Data for the Evaluation Scenarios}
\label{app:scenarios-data}

\begin{figure}[b]
  \centering
  \includegraphics[width=\columnwidth]
  {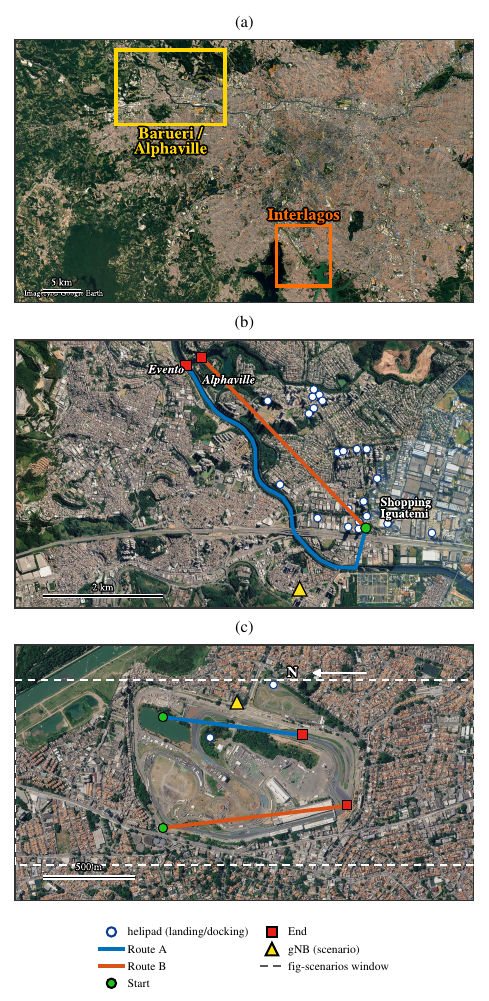}
    \caption{Public data used to ground the evaluation scenarios.
(a) S\~ao Paulo metropolitan region and the two study areas.
(b) Barueri/Alphaville delivery corridors, with licensed
5G NR stations from one operator and registered helipads.
(c) Interlagos circuit and the map region used in
Fig.~\ref{fig:scenarios}.}
  \label{fig:real_scenarios}
\end{figure}

The scenarios of Fig.~\ref{fig:scenarios} are controlled models
constructed over real geographies, rather than replays of measured
flights or cellular traces. Figure~\ref{fig:real_scenarios} shows
the public infrastructure and operational data used to select the
study areas and constrain their geometry. 

We obtain the coordinates and technology of licensed base
stations from Anatel's public Mosaico/Spectrum-E database
\cite{anatel_mosaico_stations}. Because the database contains
separate records for the emissions and sectors associated with
the same station, we deduplicate entries by station number.
Figure~\ref{fig:real_scenarios} shows only licensed 5G NR
stations from a single operator, TIM. We also obtain registered
helipad locations from ANAC. These data establish the presence
of cellular and aviation infrastructure around both study
areas; they are not used as measurements of aerial coverage.
Radio conditions are generated by the channel model described
in \S~\ref{sec:eval-methodology}. To identify a suitable experimental gNB for the single-cell scheduling study, we evaluated all licensed 5G NR stations in the region and selected the tower at the highest elevation.

\noindent\textbf{Delivery corridors.}
Since June 2026, Speedbird Aero and iFood have operated an
ANAC/DECEA-authorized BVLOS delivery service from Iguatemi
Alphaville to nearby residential complexes. The reported operation
uses a 3.6-km aerial leg, flown in approximately five minutes at up
to 60\,m AGL and 50\,km/h, carrying payloads of up to 5\,kg
\cite{mundogeo26speedbird,ifood26drone}. Our simulated corridors
use 4.0--4.9-km trajectories at 60\,m and 15\,m/s, preserving the
distance, altitude, and speed range of the reported operation.

The commercial mall-to-droneport operation motivates the endpoint
structure of the delivery scenario. We then construct two routes
over the same region with different combinations of radio quality
and population exposure. This controlled contrast lets us isolate
whether a scheduler follows the channel conditions or the
operational consequence of delaying each UAV.

\noindent\textbf{Crowded-event scenario.}
During the Formula~1 S\~ao Paulo Grand Prix, DECEA activates
\emph{Opera\c{c}\~ao Lagos} to coordinate aerial activity around
Interlagos. In 2025, the operation handled 545 movements over three
days, including 427 helicopter movements and 118 authorized drone
flights \cite{mundogeo25lagos,decea24lagos}. This observation
motivates an evaluation scenario in which UAV traffic and ground
exposure are concentrated within a short operating window.

Our simulated trajectories follow the 4.3-km Interlagos circuit,
while spatially nonuniform spectator regions create different
consequences for stale information along the route. The simulation
does not replay the reported flights or claim to reproduce the
event's traffic management procedures. Instead, it uses the
documented coexistence of dense crowds and aerial activity to
construct a controlled scheduling workload.

Together, these data ground the scenarios in real locations and
documented operating ranges, while the controlled channel,
population, and traffic models isolate the scheduling behavior
studied in this paper.

\section{Evaluation Across the \(K\times N\) Grid}
\label{app:deep-ra}

\begin{figure*}[t]
    \centering
    \includegraphics[width=\textwidth]
    {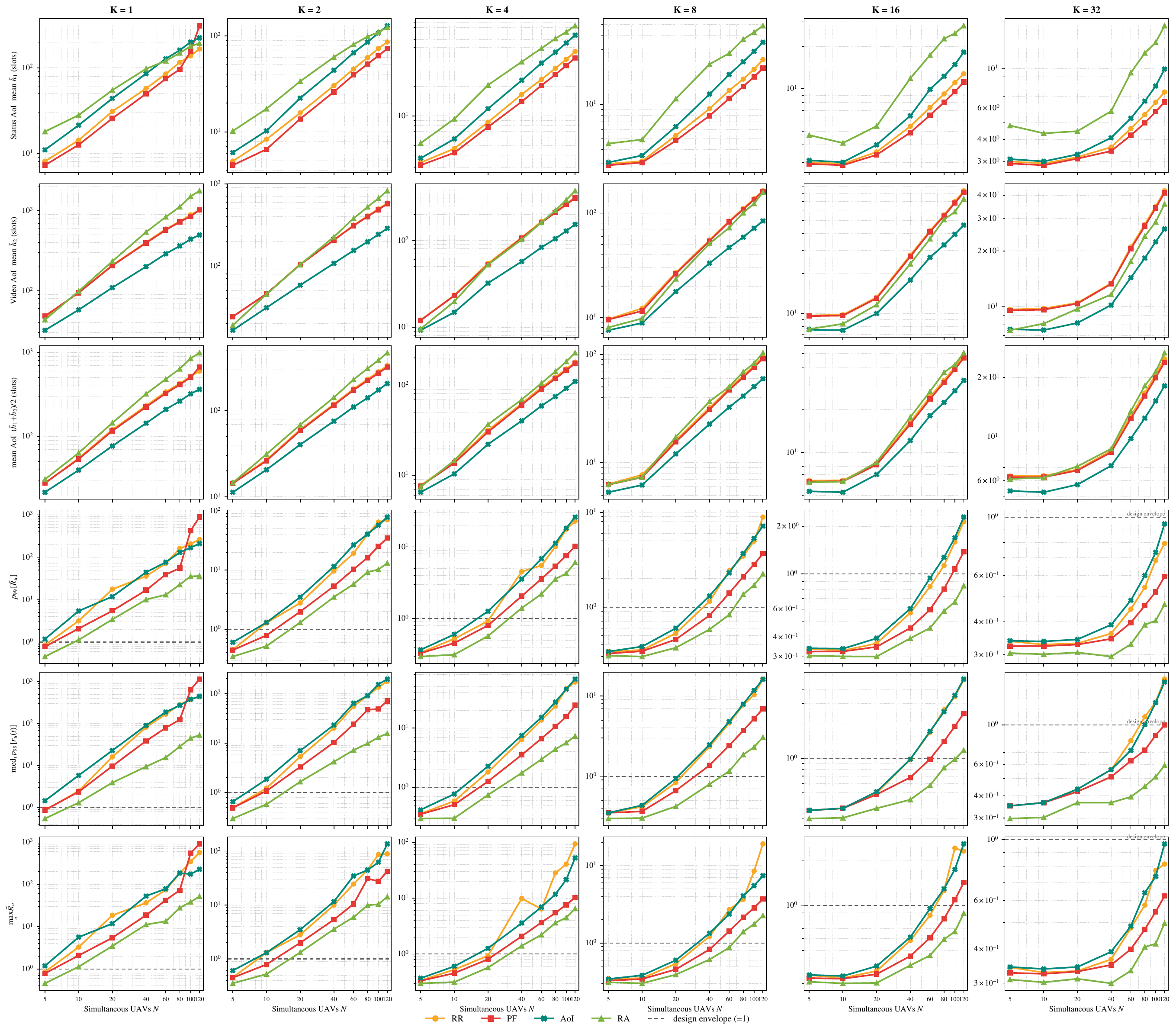}
    \caption{Complete evaluation across all 48 operating points.
    Columns represent the available uplink channels
    \(K\in\{1,2,4,8,16,32\}\), and each panel varies the number
    of concurrent UAVs \(N\in\{5,\ldots,120\}\). The first two
    rows report fleet-mean status and video AoI, and the third
    row their combined mean. The remaining rows report the 99th
    percentile of the per-UAV risk, the median instantaneous
    tail, and the worst-UAV risk. The dashed line marks the
    operational envelope \(R=1\).}
    \label{fig:deep-ra}
  \end{figure*}

  Figure~\ref{fig:deep-ra} reports the individual operating points
  underlying the load-regime aggregates of
  Fig.~\ref{fig:risk-cost}. For each policy, the three risk rows show
  \[
  p_{99,u}[R_u],\quad
  \operatorname{median}_{t}\!\left[p_{99,u}[R_u(t)]\right],
  \quad\text{and}\quad
  \max_u R_u .
  \]
  These views distinguish persistent per-UAV violations of the
  operational envelope from sustained instantaneous tails.

  \noindent\textbf{Only RA remains near the envelope under contention.}
  At \(K=8\) and \(N=60\), every baseline breaches all three risk
  views: the worst UAV reaches \(1.4\times\) the boundary for PF,
  \(2.4\times\) for AoI, and \(2.7\times\) for RR, with per-UAV
  99th percentiles of \(1.40\), \(2.30\), and \(2.42\). RA produces risk values of \(0.83\), \(1.16\), and \(0.87\) across the three rows: it stays inside the envelope on both per-UAV views and within \(1.2\times\) on the sustained instantaneous tail, at an operating point where no baseline satisfies any of the three.

  \noindent\textbf{The advantage persists under a stricter envelope.}
  Using \(\max_uR_u\leq1\) as the admission criterion, RA
  supports \(N=60\) at \(K=8\) and the entire evaluated range up to
  \(N=120\) at \(K=16\). PF supports \(N=40\) and \(N=80\) at the
  same channel budgets, while RR and AoI support \(N=20\) and
  \(N=60\). This worst-UAV frontier coincides with the
  \(p_{99}\)-based frontier at every \(K\): for RA, the fleet-level
  guarantee comes at no additional admission cost. At \(K=32\),
  all policies remain within the envelope throughout the range.

  \noindent\textbf{Risk reduction does not require uniformly lower AoI.}
    At \(N=120\), RA's status AoI ranges from \(0.9\times\) that of the
        AoI-optimal policy at \(K\leq2\) to \(1.8\times\) at \(K=32\), while
        its video AoI ranges from \(1.4\times\) at \(K=32\) to \(3.6\times\)
        at \(K=1\): under heavy contention RA spends its transmission
        opportunities on status updates, whose staleness consumes the risk
        margin fastest. RA still improves video freshness over the
        network-centric baselines under moderate contention: at \(K=8\) and
        \(N=60\), its mean video AoI is \(72\) slots, compared with \(83\)
        for PF and \(84\) for RR. RA therefore changes where the
  available freshness is placed, rather than attempting to
  minimize every age equally.

\noindent\textbf{Scheduling remains bounded by capacity.}
  As the channel budget decreases, every policy eventually exceeds
  the operational envelope: for \(K\leq4\), none keeps the worst
  UAV inside it beyond \(N=20\). At \(K=1\) and \(N=120\), even RA
  cannot compensate for the service deficit, and the policies exceed the worst-UAV boundary by \(52\times\) (RA) to
        over \(900\times\) (PF). Conversely, at \(N=5\), the policies converge because
  contention is limited; only the AoI baseline still breaches the
  envelope there (\(1.18\) at \(K=1\)). Scheduling can reshape the
  risk tail only within the feasibility region established by the
  available radio capacity.

\section{SAIL Assignment and Allocation Audit}
\label{app:sail-allocation}

This appendix supports Section~\ref{sec:eval-sora} by documenting
the SAIL assignment, audit metrics, and controlled experiment.
SAIL is assigned offline and is never exposed to any scheduler.

\noindent\textbf{SAIL assignment.}
SORA combines the final Ground Risk Class (GRC) and residual Air
Risk Class (ARC) to determine SAIL
\cite{sora2024,jarus_sora25_annexf}.
Ground-risk classification uses the maximum population density
within the intrinsic GRC footprint, comprising the operational
volume and ground-risk buffer. Thus, crossing a crowd affects
the classification of the operation, including its low-exposure
segments.

For the controlled audit, we use the aircraft column corresponding
to a maximum characteristic dimension of \(3\)~m and maximum speed
of \(35\)~m/s, assume residual ARC-b, and claim no ground-risk
mitigation credit. The assumed GRCs are \(3\) for the low-exposure
operation and \(6\) for the crowd-exposed operation, yielding
SAIL-II and SAIL-V, respectively. These labels remain fixed during
each flight. The per-slot partition in Section~\ref{sec:eval-sora}
therefore groups the active UAVs by their operation labels.

RA uses the population exposure within its staleness-dependent
uncertainty region, as defined in~\eqref{eq:ground-risk}.
The audit tests whether the resulting allocation aligns with
SORA's operation-level classification. This is a consistency check
between two uses of operational exposure; it does not establish
regulatory compliance.

\noindent\textbf{Audit metrics.}
We retain the full-occupancy window, excluding fleet fill and drain.
The allocation metric is \(\Gamma_{\mathrm{sched}}\), defined in
Section~\ref{sec:eval-sora}: each UAV's scheduling probability is
computed over its observed slots in the corresponding class, and
the ratio compares the class means. A value of one denotes equal
per-UAV scheduling probability.

At the decision level, \(P_{\mathrm{align}}\) is the fraction of
matched comparisons in which the SAIL-V UAV is served. Candidates
must belong to different classes, both be in coverage, differ by
at most one slot in status AoI and \(3\)~dB in SNR, and have exactly
one UAV scheduled. The reference \(P_{\mathrm{align}}=0.5\)
denotes no preference within these comparisons. We exclude
operating points with fewer than \(50\) matched comparisons.
Matching controls status freshness and SNR within the stated
tolerances; it does not equate every scheduler state. We interpret
this metric as evidence of allocation alignment under these
controls. Each retained operating point contributes one sample
to the corresponding panel; medians summarize operating points.

\noindent\textbf{Controlled scenario.}
\label{app:sail-controlled}
A single gNB serves an low-exposure corridor and a
crowd-exposed corridor. Each is flown in both directions,
with departures equally divided among the four routes. Links
remain in coverage, with SNRs between \(12\) and \(27\)~dB.
The crowd-exposed corridor contains \(24\) hotspots at eight
longitudinal positions, each with a central cluster and two
lateral clusters offset by \(150\)~m. The hotspots occupy the
central \(3.1\)~km, leaving \(600\)~m and \(300\)~m clear at the
two endpoints. Their peak density is \(\rho_{\max}\).

The experiment uses the main evaluation's traffic, delivery, and
policy parameters. Corridor lengths and channel histories differ;
the decision audit therefore compares candidates satisfying the
matching conditions above. We exclude \(N<20\), where too few
contended decisions are available.

\begin{figure}[t]
  \centering
  \includegraphics[width=\columnwidth]
    {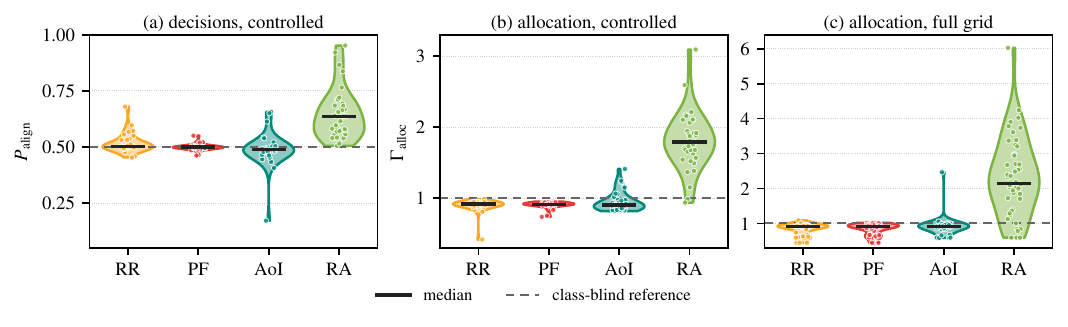}
  \caption{Allocation audit. Controlled scenario:
  (a)~decision alignment \(P_{\mathrm{align}}\) and
  (b)~allocation ratio \(\Gamma_{\mathrm{sched}}\).
  (c)~Allocation across the \(48\) evaluation-grid operating points,
  including the six with \(K\geq N\). Points represent operating
  points; black bars indicate medians.}
  \label{fig:sail-composite}
\end{figure}

\noindent\textbf{Results.}
\label{app:sail-results}
In Fig.~\ref{fig:sail-composite}(a), RA serves the SAIL-V candidate
in \(64\%\) of matched comparisons on median, reaching \(95\%\)
at the most contended points. RR, PF, and AoI yield medians of
\(50\%\), \(50\%\), and \(49\%\), respectively. Thus, RA's
preference persists among candidates with comparable status AoI
and SNR.

The allocation metric shows the same direction (panel~b).
RA provides a median \(1.8\times\), and up to \(3.1\times\),
the per-UAV scheduling probability to SAIL-V operations.
Baseline ratios lie between \(0.90\) and \(0.91\); we treat this
departure from one as scenario-induced variation.

Panel~c extends the allocation audit to the evaluation grid
reported in Section~\ref{sec:eval-sora}. Across the \(42\) points with \(N>K\), aggregated across both
        evaluation scenarios, RA's median ratio is \(2.37\), compared with
        \(0.92\) for each baseline.
The six points with \(K\geq N\) are included here but excluded
from Fig.~\ref{fig:sail-gamma-violin}, since allocation differences
primarily reflect traffic availability in that regime.
Together, these results show that RA directs more service toward
higher-SAIL operations without receiving their labels, consistent
with the allocation behavior discussed in the main text.

\section{Sensitivity to Video-Frame Size \(L\)}
\label{app:l-sweep}

The analysis in Sections~\ref{sec:stale-to-risk}
and~\ref{sec:policy} models video frames with a fixed service
requirement \(L\). Our main evaluation instead uses a realistic 
workload \cite{aviator2021}: one 30-packet I-frame followed by fourteen
3-packet P-frames, yielding a mean frame size of \(4.8\) packets .
To isolate the effect of frame size, we repeat the complete
\(K\times N\) evaluation with constant
\(L\in\{5,10,20,30\}\), keeping all other parameters unchanged.

\noindent\textbf{Group of Pictures (GOP) versus constant-size frames.}
Constant \(L=5\) closely matches the GOP mean. Relative to the
GOP workload, only 19 of the 192 policy--operating-point pairs
change sides of the risk envelope, and none changes the policy
ordering. Removing the I-frame bursts lowers median
\(p_{99,u}[R_u]\) by \(29\%\) for RR and AoI, \(22\%\)
for PF, and \(12\%\) for RA. Thus, GOP burstiness makes the
workload more demanding, but does not determine RA's advantage.

\begin{figure}[t]
  \centering
  \includegraphics[width=\columnwidth]
  {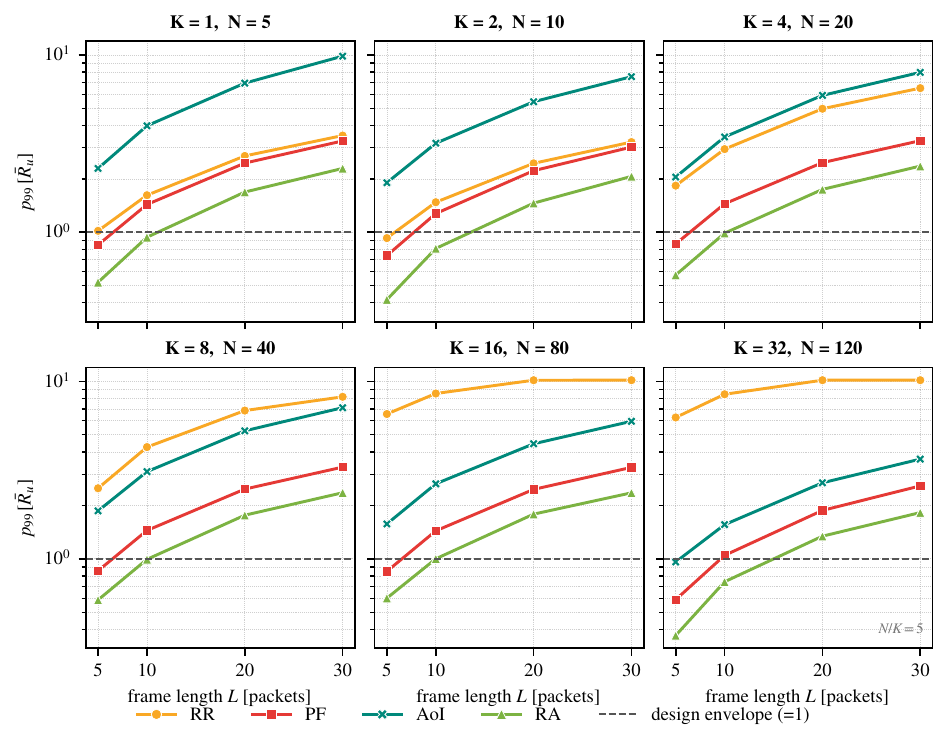}
  \caption{Tail risk as the constant video-frame size \(L\)
  increases at the design-load operating points \(N/K=5\).
  Larger frames increase risk for every policy, but the ranking
  remains unchanged. RA is the last policy to leave the
  operational envelope.}
  \label{fig:lseries-dose}
\end{figure}

\noindent\textbf{Sensitivity to \(L\).}
Figure~\ref{fig:lseries-dose} shows that increasing \(L\) from
5 to 30 packets raises median tail risk by \(2.6\times\) for RA,
\(2.9\times\) for PF, \(3.6\times\) for AoI, and \(1.9\times\)
for RR. RR's smaller relative increase reflects that it is
already highly exposed at small \(L\), rather than greater
resilience. RA remains inside the envelope at \(L=5\), with
\(p_{99,u}[R_u]\) between \(0.37\) and \(0.60\), and
reaches the boundary at \(L=10\), with values between \(0.74\)
and \(1.00\). PF crosses the boundary between \(L=5\) and
\(L=10\), while RA crosses only for \(L>10\).

Larger updates increase risk for every scheduler because they
occupy more transmission opportunities. However, the policy
ordering remains unchanged over a \(6\times\) increase in \(L\):
RA consistently retains more operational margin and is the last
policy to leave the envelope.

\section{From Radio to Air Capacity}
\label{app:operational-interface}

Communication-constrained operational capacity provides a natural interface between the cellular network and airspace traffic management. When projected simultaneous demand approaches the capacity supported by a region, traffic management mechanisms can reshape that demand before communication becomes the binding constraint. Possible actions include delaying departures, adjusting speeds, rerouting selected operations, or deferring additional aircraft until sufficient communication margin becomes available. The network does not authorize an operation; rather, it exposes when expected traffic exceeds the capacity that available communication resources can support within the declared envelope.

The same interface can also inform infrastructure planning. Expansion, improved aerial coverage, additional sectors, or new base stations can be interpreted through both conventional throughput gains and the additional low-altitude operations they enable within the same communication-risk constraint. Regions that repeatedly approach this boundary therefore identify where network upgrades or operational traffic shaping may provide the greatest benefit.

This interpretation is intentionally limited to communication-induced exposure. Within this scope, RA complements the broader aviation safety case, which also encompasses SORA, airspace constraints, weather assessment, vehicle reliability, and other operational factors. RA expresses communication capacity in a form that traffic-management and safety-assessment processes can consume. Extending this interface to multi-cell systems, handover regions, and city-scale traffic management remains future work.

\section{Simulation Toolbox Parameters}
\label{app:toolbox}

We implement the slot-level simulator in MATLAB using UAV Toolbox~\cite{toolboxuav} for mobility and 5G Toolbox~\cite{toolbox5g} for radio propagation.

\noindent\textbf{Mobility.}
Each UAV follows a \texttt{waypointTrajectory} along a
straight corridor, climbing from a $15$\,m vertiport to $60$\,m over
$500$\,m and descending symmetrically. UAVs travel at $15$\,m/s in
the evaluation and $15$\,m/s in Fig.~\ref{fig:scenarios}. Departures are
uniformly spaced, while traffic-generation phases are randomized.
For scalability, the evaluation uses an equivalent closed-form
trajectory engine, cross-checked against UAV Toolbox slot by slot.

\noindent\textbf{Channel.}
We use the 3GPP TR~38.901 UMa line-of-sight model
(\texttt{nrPathLoss}), with a $3.5$\,GHz carrier,
$20$\,MHz bandwidth, $30$\,m gNB, $23$\,dBm UAV transmit power,
$15$\,dBi receiver gain, and $7$\,dB noise figure. The evaluation
uses deterministic path loss and one serving gNB. The motivating
scenario additionally applies $4$\,dB log-normal fading in blocks of
$50$ slots, using the same realization across policies.

\noindent\textbf{Delivery.}
Each of the $K$ channels carries one packet per slot. A scheduled
packet succeeds independently with probability $0.8$ when
$\mathrm{SNR}\geq5$\,dB and is lost otherwise. The Shannon rate is used
only by PF. We add a fixed $10$\,ms one-way radio, core, and
processing delay.

\end{document}